\documentclass[%
 r eprint,superscriptaddress,
 amsmath,amssymb,
 aps,
pra,twocolumn
]{revtex4-2}

\usepackage{graphicx}
\usepackage{dcolumn}
\usepackage{bm}
\usepackage{amsthm} 
\usepackage{hyperref}
\hypersetup{
    colorlinks=true,
    linkcolor=blue,
    filecolor=magenta,      
    urlcolor=cyan,
    pdftitle={Overleaf Example},
    pdfpagemode=FullScreen,
    }

\usepackage{algorithm}
\usepackage{algpseudocode}
 \makeatletter
  \renewcommand{\ALG@name}{Protocol}
  \makeatother

\usepackage{bm}
\usepackage{xcolor}
\usepackage{amsmath,amssymb}
\usepackage{caption}
\usepackage{subcaption}
\usepackage[shortlabels]{enumitem}
\usepackage{mdframed}

\usepackage{placeins} 

\usepackage{tikz}

\usetikzlibrary{calc, positioning, arrows.meta}
\definecolor{boxBorder}{HTML}{9E9E9E}
\definecolor{measGreenLight}{HTML}{E8F5E9}
\definecolor{measGreen}{HTML}{2E7D32}
\definecolor{nodeDark}{HTML}{212121}
\definecolor{edgeGray}{HTML}{757575}
\definecolor{crimsonRed}{HTML}{C62828}

\tikzset{
  stepBox/.style={
    rounded corners=8pt,
    draw=boxBorder,
    line width=1.1pt,
    fill=white,
    minimum width=4.6cm,
    minimum height=3.8cm
  },
  stepNum/.style={
    font=\bfseries\normalsize,
    text=nodeDark,
    inner sep=0pt,
    anchor=north east
  },
  captionText/.style={
    font=\bfseries\footnotesize,
    text=nodeDark,
    align=center,
    text width=4.2cm
  },
  pics/meter/.style={
    code={
      \draw[rounded corners=1.5pt, fill=measGreenLight, draw=measGreen, line width=0.8pt]
        (-0.25,-0.18) rectangle (0.25,0.18);
      \draw[line width=0.6pt, draw=measGreen!80!black]
        (-0.17, 0.01) arc (150:30:0.2);
      \draw[line width=0.8pt, draw=measGreen!90!black, -{Latex[scale=0.6]}]
        (0,-0.1) -- (0.09, 0.12);
      \fill[measGreen!90!black] (0,-0.1) circle (0.7pt);
    }
  },
  gnode/.style={circle, fill=nodeDark, inner sep=0pt, minimum size=6pt},
  baseEdge/.style={line width=1.1pt, draw=edgeGray},
  treeEdge/.style={line width=2.8pt, draw=measGreen},
  flowArrow/.style={-{Latex[length=3.8mm, width=2.5mm]}, line width=1.6pt, draw=nodeDark!80}
}

\usetikzlibrary{quantikz2}
\usetikzlibrary{graphs}
\usetikzlibrary{fit}

\usepackage{bbm}
\usepackage{physics}
\usepackage{cleveref}

\newtheorem{theorem}{Theorem}
\newtheorem{remark}{Remark}
\newtheorem{lemma}{Lemma}
\newtheorem{proposition}{Proposition}

\newtheorem{corollary}{Corollary}

\DeclareMathOperator{\perm}{perm}

\newcommand{\U}[1]{\widehat{\mathcal{#1}}}

\renewcommand{\H}{\mathcal{H}}

\newcommand{\id}{\mathbbm{1}}

\newcommand{\E}{\mathbb{E}}
\newcommand{\hext}{\H_{\mathrm{ext}}}
\newcommand{\hint}{\H_{\mathrm{int}}}

\newcommand{\Omegaext}{\Omega_{\mathrm{ext}}} 
\newcommand{\TVD}{\mathrm{TVD}} 

\newcommand{\bias}{\mathrm{Bias}}
\newcommand{\variance}{\mathrm{Var}}
\newcommand{\equivclass}{\mathcal{G}}

\usepackage{xcolor}
\definecolor{myDarkBlue}{RGB}{0, 0, 139}

\definecolor{darkgreen}{HTML}{006400}

\preprint{APS/123-QED}

\begin{document}

\begin{abstract}
Partial distinguishability affects multiphoton interference, which is central to large-scale photonic experiments such as Boson Sampling and to applications of multiphoton states across quantum technologies. Characterizing the internal degrees of freedom responsible for this behavior is challenging: in many scenarios they cannot be accessed directly and even their dimension or physical nature may be unknown \emph{a priori}. Here we introduce Gram-matrix tomography for independently prepared photons with pure or nearly pure internal states. For pure states, the Gram matrix contains all information required to predict the statistics of experiments that are insensitive to the internal degrees of freedom. We reconstruct this matrix, up to physically irrelevant phase choices, from a quadratic number of measurable Bargmann invariants (quantities encoding overlap magnitudes and collective phases). We access these quantities using two complementary approaches: targeted Fourier interferometers and two- and three-body photon-number correlations measured in a single random interferometer. We derive explicit error bounds that control the total-variation distance between predicted and actual outcome distributions for any subsequent experiment that is insensitive to the internal states. These guarantees account for finite sampling and deviations from purity. Crucially, the total sample complexity of the protocol is independent of the internal Hilbert-space dimension and polynomial under mild assumptions on overlaps of internal states. Our protocol uses standard interferometric  measurements and turns them into a robust and operationally grounded characterization of partially distinguishable photons.
\end{abstract}

\title{Characterizing Photonic Partial Distinguishability\\ Through Gram-Matrix Tomography} 

\author{Marcin Kotowski}
\affiliation{Center for Quantum Enabled-Computing, Center for Theoretical Physics of the Polish Academy of Sciences, Al. Lotników 32/46, 02-668 Warsaw
, Poland}

\author{Marco Robbio}
\affiliation{International Iberian Nanotechnology Laboratory (INL), Av. Mestre Jos\'e Veiga, 4715-330 Braga, Portugal}
\affiliation{Centre for Quantum Information and Communication, \'Ecole polytechnique de Bruxelles, CP 165/59, Universit\'e libre de Bruxelles, 1050 Brussels, Belgium}

\author{Oliver Reardon-Smith}
\affiliation{Center for Quantum Enabled-Computing, Center for Theoretical Physics of the Polish Academy of Sciences, Al. Lotników 32/46, 02-668 Warsaw
, Poland}

\author{Benoit Seron}
\affiliation{International Iberian Nanotechnology Laboratory (INL), Av. Mestre Jos\'e Veiga, 4715-330 Braga, Portugal}

\author{Leonardo Novo}
\affiliation{International Iberian Nanotechnology Laboratory (INL), Av. Mestre Jos\'e Veiga, 4715-330 Braga, Portugal}

\author{Ernesto F. Galvão}
\affiliation{International Iberian Nanotechnology Laboratory (INL), Av. Mestre Jos\'e Veiga, 4715-330 Braga, Portugal}
\affiliation{Instituto de F\'isica, Universidade Federal Fluminense, Av. Gal. Milton Tavares de Souza s/n, Niter\'oi, RJ, 24210-340, Brazil}

\author{Micha{\l} Oszmaniec}
\affiliation{Center for Quantum Enabled-Computing, Center for Theoretical Physics of the Polish Academy of Sciences, Al. Lotników 32/46, 02-668 Warsaw
, Poland}

\maketitle


\section{Introduction}
\label{sec:introduction}

Many photonic quantum technologies rely on interference between photons whose internal states cannot be perfectly controlled. Unwanted partial distinguishability can introduce errors in linear-optical quantum gates~\cite{rohdeModeMismatch2006}, reduce the sensitivity of quantum-enhanced interferometry~\cite{jachuraModeEngineering2016} and reduce quantum computational advantage schemes such as Boson Sampling to efficient classical computations ~\cite{renemaPartialDistinguishability2018}. Characterizing these effects is therefore essential for assessing photonic experiments and improving photonic devices. Yet the internal degrees of freedom responsible for them are often inaccessible: photons may differ in polarization, frequency, or arrival time, while the available interferometers and detectors manipulate and resolve only spatial modes. The observed interference patterns carry information about the internal states, but extracting that information requires measurements and reconstruction methods adapted to this restricted access.

The theoretical description of partial distinguishability was developed in seminal works ~\cite{tichySamplingPartiallyDistinguishable2015,shchesnovichPartialIndistinguishabilityTheory2015a}. In these works it was observed that, although Hong--Ou--Mandel visibilities characterize pairwise interference, multiphoton experiments also depend on relational information involving larger groups of photons. For independent preparations, this information is encoded in \emph{Bargmann invariants}, the multivariate traces $\operatorname{Tr}(\rho_1\cdots\rho_k)$ of the individual internal states $\{\rho_i\}$ ~\cite{oszmaniecMeasuringRelationalInformation2021}. For pure states, these quantities are cyclic products of inner products; their arguments are the collective phases that remain unchanged when individual state vectors are rephased~\cite{shchesnovichCollectivePhasesIdentical2018}. Experiments have demonstrated how such phases alter multiphoton interference without changing pairwise visibilities~\cite{menssenDistinguishabilityManyParticleInterference2017,jonesInterferingDistinguishablePhotons2020}, and how mixedness introduces further features beyond the two-photon description~\cite{jonesDistinguishabilityMixednessQuantum2023}.

These findings have motivated the development of experimentally accessible ways to characterize multiphoton sources without resolving their internal states. Indicators of genuine multiphoton indistinguishability and their experimental implementations quantify the collective character of interference within specified incoherent models for the photon sources~\cite{Brod_GI,giordaniExperimentalQuantification2020,pontQuantifyingPhotonIndistinguishability2022}, while bunching and photon-number fluctuations provide further tests that tolerate imperfect control of the interferometer~\cite{rodariSemideviceIndependentCharacterization2024}. More recently, protocols using Fourier interferometers have reduced the measurement cost of such benchmarks~\cite{sanzExponentialImprovementBenchmarking2026,schadow2026certification} for restricted classes of input states,  and finally  operational certification methods have established guarantees on the proximity to perfectly indistinguishable preparations~\cite{OptimalInterferometricCertification}. These methods can certify whether multiple photons from a source are sufficiently indistinguishable for a given purpose. A complementary task, which has so far attracted little attention, is to fully recover the information needed to describe partially distinguishable photons themselves, including their collective phases, and to quantify the accuracy of the resulting predictions.

Rigorous quantum-state tomography provides a model for this task: a reconstruction of a quantum state should come with a quantitative relation between the number of collected samples and the estimation error. For finite-dimensional systems, including registers of distinguishable qubits and qudits, sample-complexity bounds quantify the resources needed to achieve a prescribed trace-distance accuracy~\cite{haahSampleOptimalTomography2017,odonnellEfficientQuantumTomography2016}. Low-rank tomography and compressed sensing further exploit structure to reduce measurement requirements while controlling statistical error~\cite{grossQuantumStateTomography2010,flammiaQuantumTomographyCompressed2012}. Similar guarantees have been developed for energy-constrained continuous-variable systems and, in particular, bosonic Gaussian states~\cite{meleLearningQuantumStates2025,bittelOptimalEstimatesTrace2025}. 

\begin{figure*}[htbp!]
    \centering
    \includegraphics[width=\textwidth]{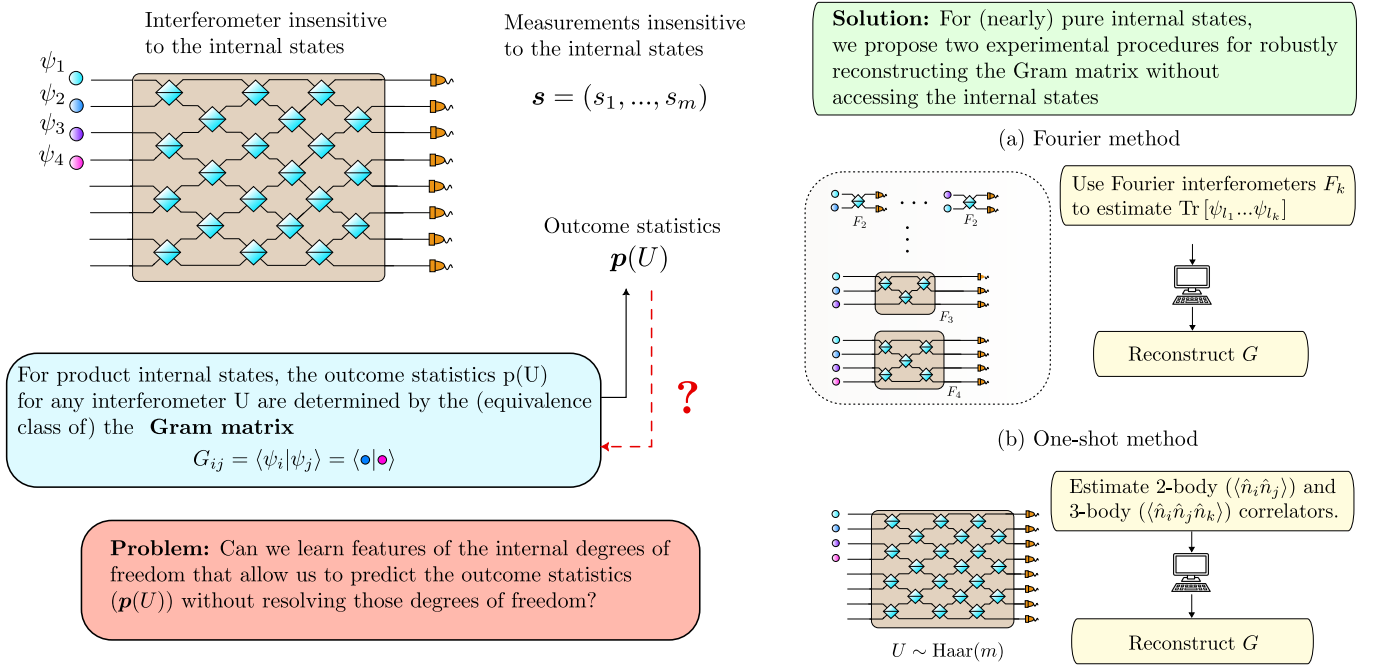}
    \caption{A pictorial description of the main problem studied  together with the key results of the paper. \textbf{Left:} For $n$ single photons with pure internal states $|\psi_1\rangle, \dots, |\psi_n\rangle$, all relational information governing multiphoton interference is encoded in the $n \times n$ Gram matrix $G_{ij} = \langle \psi_i | \psi_j \rangle$. The tomographic objective is to reconstruct $G$ without resolving the internal degrees of freedom. \textbf{Right:} Two complementary experimental protocols to extract the needed Bargmann invariants: (a)~\textit{Fourier method}: Uses $k$-mode Fourier interferometers $F_k$ to directly measure cyclic invariants $\mathrm{Tr}[\psi_{k_1}\dots\psi_{k_l}]$, see Protocol~\ref{Protocol: Fourier}.  (b)~\textit{One-shot method}: Collects output samples from a single Haar-random interferometer $U \sim \mathrm{Haar}(m)$ to estimate photon-number 2-body and 3-body correlators, obtaining the invariants via linear inversion, see Protocol~\ref{Protocol: One-shot}. Both protocols reconstruct an estimate $\widehat{G}$ through classical postprocessing. }
\label{fig:overview}

\end{figure*}

For partially distinguishable photons, the corresponding target is the \emph{external state}~\cite{adamsonMultiparticleStateTomography2007,adamsonDetectingHiddenDifferences2008,englbrechtIndistinguishabilityIdenticalBosons2024,OptimalInterferometricCertification}. It is obtained by tracing out the internal degrees of freedom from the physical multiphoton state, leaving a density operator on the external, or path, degrees of freedom. This state, also called the visible state or accessible state, determines the statistics of every experiment whose operations and measurements are insensitive to the internal degrees of freedom. Consequently, reconstructing it in trace distance controls the total-variation error of all such measurement statistics. Existing work established that reconstruction of internal degrees of freedom can be carried out in principle~\cite{oszmaniecMeasuringRelationalInformation2021} from estimates of Bargman invariants. What has been missing in this setting is a reconstruction guarantee that connects finite-data estimates  obtained without resolving the internal degrees of freedom to the trace-distance accuracy of the associated external state. This reconstruction problem and our two experimental approaches are illustrated in Fig.~\ref{fig:overview}.

In this work, we fill in this gap by developing  a reconstruction theory for independently prepared photons with pure or nearly pure internal states, based solely on linear optics and photon-number resolving detectors. The protocols are furthermore efficient (in number of photons, irrespectively of the dimension and character of internal degrees of freedom) under explicit assumptions on the internal state overlaps. For $n$ photons entering distinct input modes with pure internal states $\ket{\psi_1},\ldots,\ket{\psi_n}$, all relevant information is contained in their Gram matrix
\begin{equation}
    G_{ij}=\langle\psi_i|\psi_j\rangle,
    \qquad G\succeq0,\qquad G_{ii}=1 .
    \label{eq:intro-gram}
\end{equation}
In particular, the external state, which specifies the outcome probabilities of linear optical experiments with photon number detectors, is determined by $G$. Independent phase choices for the vectors transform $G$ by conjugation with a diagonal unitary and leave the external state unchanged. Our tomography procedure reconstructs $G$ up to these unphysical phase choices from $O(n^2)$ selected Bargmann invariants (Sec.~\ref{sec:protocol-main}). For pure inputs, our Fourier measurement method achieves Frobenius error $\|G-\widehat G\|_{\mathrm F}\leq\delta$ at fixed confidence using $O(n^4/\delta^2)$ total experimental runs and polynomial classical postprocessing, provided the required pairwise and cycle invariants have magnitudes bounded below by a positive constant (Proposition~\ref{th:main-MSE-simple}). More generally, we give explicit bounds retaining the dependence on these invariants and on the measurement variances, but which are completely independent of the internal Hilbert-space dimension (Appendix~\ref{sec:estimators-error}).

The reconstruction guarantee has a direct operational meaning. We show that bosonic symmetrization acts as a mechanism translating a reconstruction for the internal state Gram matrix $G$ [Eq. (\ref{eq:intro-gram})] into a reconstruction of the state $\Omegaext(G)$ of external degrees of freedom, leading to a bound on the total variation distance (TVD) of \textit{any} linear-optical experiment insensitive to internal degrees of freedom (Proposition~\ref{prop:tvd-continuity} and Corollary~\ref{cor: total variation distance}):
\begin{align}
    \TVD\!\left(p_{\mathcal M}(G),p_{\mathcal M}(\widehat G)\right)
    &\leq \frac12\big\|\Omegaext(G)-\Omegaext(\widehat G)\big\|_1\notag\\
    &\leq \frac{\sqrt n}{2}\big\|G-\widehat G\big\|_{\mathrm F},
    \label{eq:intro-operational}
\end{align}
where $p_{\mathcal M}$ is the outcome distribution of \emph{any} measurement on the external degrees of freedom and the Gram matrices are compared after fixing their phase choices. Crucially, this supplies the same type of operational accuracy guarantee as trace-distance tomography for a system of distinguishable particles, as the external state $\Omegaext(G)$ is the relevant object for photonic experiments. The preceding sample bound gives $O(n^5/\delta^2)$ runs for external-state trace-distance accuracy $\delta$ under the same assumptions (Theorem~\ref{th:protocol} and the discussion following Theorem~\ref{th:protocol2}). A complementary $\Omega(n^2/\delta^2)$ worst-case sample lower bound shows that quadratic growth with photon number is unavoidable for general external-state reconstruction, even with arbitrary collective measurements on the external degrees of freedom (Theorem~\ref{th:sample-lower-bound}).

We extend the analysis to nearly pure product preparations $\rho=\rho_1\otimes\cdots\otimes\rho_n$, separating finite-sample statistical error from the systematic effect of impurity. The reconstructed Gram matrix approximates that of the dominant pure components of the individual states. If $\operatorname{Tr}(\rho_i^2)\geq1-\varepsilon$, replacing the actual source by these pure components contributes at most $n\varepsilon$ to the external-state trace distance [Eq.~\eqref{eq:trace-inequality-external}]. Combining this approximation with our reconstruction bounds gives operational guarantees that remain meaningful for imperfectly pure sources (Theorem~\ref{th:protocol2}).

To acquire the reconstruction data, we develop two complementary strategies using passive linear optics and photon counting (Fig.~\ref{fig:overview}). The first uses Fourier interferometers to measure the selected Bargmann invariants directly~\cite{novoNativeLinearopticalProtocol2026}, including higher-order invariants when the overlap structure requires them (Protocol~\ref{Prot: Fourier}). The second obtains second- and third-order invariants from photon-number correlations measured at the output of a single generic interferometer, provided the associated linear systems are invertible (Protocol~\ref{Prot: Oneshot}). These orders suffice whenever one reference photon has nonzero overlap with every other photon and the use of a single configuration reduces the need for experimental reconfiguration. The trade-off is that the correlations must be converted into invariants by linear inversion, whose conditioning affects the reconstruction error. We prove conditioning bounds for the second-order inversion in a dilute regime (Theorem~\ref{th:main-condition}) and assess both inversion problems numerically, including at substantially smaller mode numbers (Appendix~\ref{sec: Conditioning number}).

The remainder of the paper is organized as follows. Section~\ref{sec:partialdistTheory} introduces the description of partial distinguishability and the external state. Section~\ref{sec: Tomography} presents the reconstruction protocol, its error guarantees, and the sample-complexity lower bound. Section~\ref{sec:relational-info} describes the Fourier and one-shot measurement methods. The numerical results then compare the two methods for pure and mixed inputs, and the concluding section discusses open questions. Appendices~\ref{sec: proof of lemma}--\ref{sec: size reduction} contain the proofs and further analyses of estimation errors, inversion conditioning, photon-number correlators, indistinguishability certification, and linear-system reduction.

\section{Partial distinguishability in linear optics}
\label{sec:partialdistTheory}

Consider $n$ photons entering distinct input ports of an $m$-mode passive interferometer. The paths through the interferometer are the \emph{external} degrees of freedom, with Hilbert space $\hext\simeq\mathbb{C}^{m}$ and basis $\ket{1},\ldots,\ket{m}$. Each photon also carries \emph{internal} degrees of freedom, such as its temporal or spectral wave packet and polarization, which the interferometer does not manipulate and the detectors do not resolve. Fix an orthonormal basis $\{\ket{\alpha}\}$ of the internal Hilbert space $\hint$. The creation operator $a^{\dagger}_{j,\alpha}$ then creates a photon in mode $j$ with internal basis state $\ket{\alpha}$, as in the description of Ref.~\cite{OptimalInterferometricCertification}.

A passive interferometer described by a single-particle unitary $U\in U(m)$ mixes the modes while preserving the internal label:
\begin{equation}\label{eq:interferometer-action}
    \U{U}\,a^{\dagger}_{j,\alpha}\,\U{U}^{\dagger}
    =\sum_{k=1}^{m}U_{kj}a^{\dagger}_{k,\alpha}.
\end{equation}
The detectors measure the total occupation of each output mode, $\widehat n_j=\sum_{\alpha}a^{\dagger}_{j,\alpha}a_{j,\alpha}$. Although the internal states are neither manipulated nor observed, they determine the interference between alternative many-photon paths~\cite{tichyInterferenceIdenticalParticles2014,shchesnovichPartialIndistinguishabilityTheory2015a,tichySamplingPartiallyDistinguishable2015}.

To identify the state accessible in such experiments, we keep the external and internal tensor factors explicit. Suppose the occupied input ports are $1,\ldots,n$, and write their ordered list as $\ket{\mathbf{i}_0}=\ket{1}\otimes\cdots\otimes\ket{n}$. An internal density operator $\rho$ on $\hint^{\otimes n}$ describes the preparations associated with these input ports; the tensor factors label ports, rather than distinguishable particles. Bosonic symmetry requires us to permute the external and internal factors together. With the convention $P_{\pi}\ket{x_1,\ldots,x_n}=\ket{x_{\pi^{-1}(1)},\ldots,x_{\pi^{-1}(n)}}$, the corresponding embedding of an internal state $\ket{\chi}$ is~\cite{OptimalInterferometricCertification}
\begin{equation}\label{eq:isoDEF}
    V\ket{\chi}
    =\frac{1}{\sqrt{n!}}\sum_{\pi\in S_n}
    P_{\pi}^{\mathrm{ext}}\ket{\mathbf{i}_0}
    \otimes P_{\pi}^{\mathrm{int}}\ket{\chi}.
\end{equation}
Here the superscripts specify the factors being permuted. Since the input ports are distinct, $V$ is an isometry and $\Omega(\rho)=V\rho V^{\dagger}$ is a normalized bosonic state. This construction allows arbitrary internal states, including mixed or correlated preparations.

The \emph{ external state}, also called the \emph{visible} or \emph{accessible state}~\cite{adamsonMultiparticleStateTomography2007,adamsonDetectingHiddenDifferences2008,englbrechtIndistinguishabilityIdenticalBosons2024}, is obtained by tracing out the internal degrees of freedom. It contains all information needed to predict experiments that act only on the external modes. In the orthonormal permutation basis $\ket{\sigma}:=\ket{\sigma(1),\ldots,\sigma(n)}=P_{\sigma^{-1}}^{\mathrm{ext}}\ket{\mathbf{i}_0}$, it takes the form
\begin{equation}\label{eq:reducedstate}
\begin{aligned}
    \Omegaext(\rho)
    &:=\operatorname{Tr}_{\mathrm{int}}[V\rho V^{\dagger}]\\
    &=\frac{1}{n!}\sum_{\sigma,\tau\in S_n}
    \ket{\tau}\!\bra{\sigma}\,
    \operatorname{Tr}\!\left[\big(P_{\tau\sigma^{-1}}^{\mathrm{int}}\big)^{\dagger}\rho\right].
\end{aligned}
\end{equation}
The permutation traces in this expression are the weights that enter the theory of partial distinguishability~\cite{shchesnovichPartialIndistinguishabilityTheory2015a}.

For independent preparations, $\rho=\rho_1\otimes\cdots\otimes\rho_n$, each permutation trace factorizes over its disjoint cycles:
\begin{equation}\label{eq:cycle-bargmann-factorization}
    \operatorname{Tr}\!\left[\big(P_{\pi}^{\mathrm{int}}\big)^{\dagger}\rho\right]
    =\prod_{(i_1\,\cdots\,i_{\ell})\in\mathrm{cyc}(\pi)}
    \operatorname{Tr}(\rho_{i_1}\cdots\rho_{i_{\ell}}).
\end{equation}
The cycle traces $B_c:=\operatorname{Tr}(\rho_{i_1}\cdots\rho_{i_{\ell}})$ are \emph{Bargmann invariants}: they describe relations among the preparations and are unchanged by a common unitary transformation of the internal states~\cite{oszmaniecMeasuringRelationalInformation2021}. If every preparation is pure, $\rho_i=\ket{\psi_i}\!\bra{\psi_i}$, these quantities are cyclic products of overlaps, $B_c=G_{i_1i_2}\cdots G_{i_{\ell}i_1}$, where
\begin{equation}\label{eq:gram-matrix}
\begin{aligned}
    G_{ij}&:=\langle\psi_i|\psi_j\rangle,\\
    \operatorname{Tr}\!\left[\big(P_{\pi}^{\mathrm{int}}\big)^{\dagger}\rho\right]
    &=\prod_{i=1}^{n}G_{i,\pi(i)}.
\end{aligned}
\end{equation}
Thus the $n\times n$ Gram matrix determines the entire external state, irrespective of the dimension of $\hint$. Reconstructing it suffices to characterize all subsequent experiments that are insensitive to the internal degrees of freedom~\cite{tichySamplingPartiallyDistinguishable2015,oszmaniecMeasuringRelationalInformation2021}.

The phases of the individual vectors have no physical meaning. Replacing $\ket{\psi_i}$ by $e^{\mathrm{i}\varphi_i}\ket{\psi_i}$ transforms $G$ into $D^{\dagger}GD$, where $D=\operatorname{diag}(e^{\mathrm{i}\varphi_1},\ldots,e^{\mathrm{i}\varphi_n})$, but leaves every Bargmann invariant and hence $\Omegaext(\rho)$ unchanged. The tomographic target is therefore the equivalence class
\begin{equation}\label{eq:gram-gauge-class}
    [G]=\{D^{\dagger}GD:\ D\ \text{diagonal unitary}\}.
\end{equation}
For generic full-rank Gram matrices, this leaves $(n-1)^2$ real parameters. To choose a representative, consider the \emph{overlap graph}: vertex $i$ represents the preparation $\ket{\psi_i}$, and distinct vertices $i,j$ are joined by an undirected edge whenever $G_{ij}\neq 0$. Within each connected component, choose a \emph{spanning tree}, namely a connected subgraph that contains every vertex of that component and has no cycles. Starting from any root, choose the phases of successive state vectors so that the overlaps on all tree edges are real and positive. This fixes the phase convention within that component, up to a common phase that does not affect $G$.

Every edge outside the chosen trees closes a cycle with the unique tree path between its endpoints. The Bargmann invariant around this cycle determines the phase of that edge's overlap once the tree overlaps have been made positive. Thus pairwise invariants $|G_{ij}|^2$ determine the overlap magnitudes, while higher-order invariants provide the remaining phase information. The next section develops the reconstruction protocol and shows how its Gram-matrix error controls the trace distance between external states.

\section{The tomography protocol}\label{sec: Tomography}

\subsection{Statement of the results}\label{sec:pseudopure-setup}

We will now describe the main result of this paper -- an efficient experimental protocol that reconstructs the external state $\Omegaext(\rho)$ by measuring certain Bargmann invariants, provided $\rho$ has high purity. We first describe the ideal case where the internal state $\psi \in \hint^{\otimes n}$ is a pure product state:
\begin{equation}
    \psi = \bigotimes_{i=1}^{n}\psi_i
\end{equation}
Any fixed choice of global phases allows to identify $\psi$ with a sequence of vectors $\{\ket{\psi_i}\}_{i=1}^{n}$ with the corresponding Gram matrix $G_{ij}=\langle\psi_i|\psi_j\rangle$. Combining \eqref{eq:reducedstate} and \eqref{eq:gram-matrix}, the external state in the permutation basis $\ket{\sigma}=\ket{\sigma(1),\ldots,\sigma(n)}$ is
\begin{equation}\label{eq:omegaext-gram}
      \Omegaext(\psi) = \frac{1}{n!}\sum_{\sigma,\tau \in S_n}\ket{\tau}\bra{\sigma}\prod_{i=1}^{n}G_{\sigma(i),\tau(i)}.
\end{equation}

We propose a protocol that reconstructs an estimate of the true Gram matrix $G$ based on a number of appropriate Bargmann invariants, measured using methods described in \Cref{sec:relational-info}. The protocol is described in detail in \Cref{sec:protocol-main} and then analyzed in \Cref{sec:estimators}. 


\begin{theorem}\label{th:protocol}
    There exists a protocol that, given as input a photonic state with $n$ photons with pure product internal state $\psi$ with a Gram matrix $G$, measures $O(n^2)$ Bargmann invariants using in total $N$ experimental samples and outputs a estimate $G'$ of the Gram matrix such that for a pure product state $\psi'$ described by $G'$ we have: 
    \begin{equation}
        \mathbb{P}\left( \norm{\Omegaext(\psi)- \Omegaext(\psi')}_1 > \delta \right) \leq \frac{A(n)}{N}\cdot\frac{1}{ \delta^2}.  
    \end{equation}
    Here, the probability is taken over the sampling of Bargmann invariants and the factor $A(n)$ depends on the measurement method, number of photons $n$ and the magnitude of Bargmann invariants. In particular, the external state $\Omegaext(\psi)$ can be estimated with high probability to accuracy $\delta$ in trace distance with $O\left(\frac{A(n)}{\delta^2}\right)$ experimental samples. 
\end{theorem}

 The protocol is robust in the sense that it can handle also mixed states, provided they are sufficiently close to pure. 


\begin{theorem}\label{th:protocol2}
    The protocol from \Cref{th:protocol} works also for mixed states. Given as input a photonic state with $n$ photons with mixed product internal state $\rho=\bigotimes_{i=1}^{n} \rho_i$, with $\tr(\rho_i^2) > 1-\varepsilon$, it measures $O(n^2)$ Bargmann invariants using in total $N$ experimental samples and outputs a Gram matrix $G'$ of a pure product state $\psi'$ such that:
    \begin{equation}
        \mathbb{P}\left( \norm{\Omegaext(\rho)- \Omegaext(\psi')}_1 > \delta \right) \leq \left(\frac{A(n)}{N} +B(n)\varepsilon^2 \right)\cdot\frac{1}{ \delta^2} 
    \end{equation}
    where probability is taken over the sampling of Bargmann invariants and $A(n), B(n)$ depend on the measurement method, number of photons $n$ and the magnitude of Bargmann invariants. In particular, the external state $\Omegaext(\psi)$ can be estimated with high probability to accuracy $\delta$ in trace distance with $O\left(\frac{\max\{A(n),B(n)\}}{\delta^2}\right)$ experimental samples, assuming $\varepsilon \ll \delta$. 
\end{theorem}

\emph{Efficiency of the protocol.} The number of samples required for reconstruction to desired accuracy $\delta$ depends on the functions $A(n), B(n)$, which depend on the experimental setting. We give detailed formulas and discussion in \Cref{sec:estimators-error}. A sufficient, but not necessary, criterion for efficient (polynomial in $n$) reconstruction is that all pairwise Bargmann invariants are bounded from below by a constant independent of $n$, $B_{ij} > c$. This holds e.g. for generic random states with fixed internal dimension and we expect it to be satisfied in natural experimental settings. In that case, using the Fourier method (see \Cref{sec:relational-info}) gives $A(n)=O(n^5), B(n)=O(n^3)$, so reconstruction to accuracy $\delta$ is possible with $O(\frac{n^5}{\delta^2})$ samples provided that single photon impurities are $\varepsilon = O(\frac{\delta}{n^{3/2}})$.

It is natural to ask what are the information theoretic limits on the number of samples required to reconstruct the external state.  We show that even for pure states any protocol that with high probability correctly identifies the external state up to error $\delta$ must use at least $N = \Omega(\frac{n^2}{\delta^2})$ samples. Therefore, in certain regimes the above protocols are sample-optimal up to powers of $n$. We provide the proof in \Cref{sec:proof-of-sample-lower-bound}. 

\begin{theorem}\label{th:sample-lower-bound}
    Consider any estimator $\widehat{\rho}$ that receives $N$ copies of an input state $\rho_G$, whose internal state is a pure product state $\psi_G$ with unknown Gram matrix $G$, and can perform arbitrary collective measurements on the external degrees of freedom. If $\widehat{\rho}$ satisfies:
    \begin{equation}
        \mathbb{P}\left( \norm{\widehat{\rho} - \Omegaext(\psi_G)}_1 < \delta \right) > \frac{2}{3}
        \end{equation}
        for every $G$, then necessarily $N = \Omega\left( \frac{n^2}{\delta^2} \right)$.
\end{theorem}

\begin{figure*}[tbhp!]
        \centering
        \includegraphics[width=\linewidth]{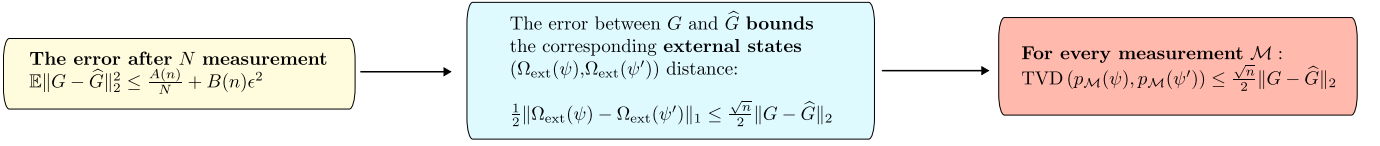}
        \caption{We consider the case of nearly pure states, i.e. $\Tr{\rho^2_{i}}\geq 1-\epsilon \ \forall i$, after $N$ measurements, the Frobenius estimation error satisfies $\mathbb{E}\|G - \widehat{G}\|_2^2 \le A(n)/N + B(n)\epsilon^{2}$ (Proposition~\ref{th:main-MSE-simple}). Additionally Proposition ~\ref{prop:tvd-continuity}, shows that being able to reconstruct the Gram matrix up to desired precision suffices to reconstruct the external state. Due to the relation between trace distance and total variation distance, the estimated Gram matrix provides a uniform guarantee for any subsequent measurement $\mathcal{M}$ on the external modes, bounding the total-variation distance of outcome distributions by $\mathrm{TVD}(p_\mathcal{M}(\psi), p_\mathcal{M}(\psi')) \le \frac{\sqrt{n}}{2}\|G - \widehat{G}\|_2$ (Corollary~\ref{cor: total variation distance}).}
        \label{fig: operational meaning}
\end{figure*}
\subsection{The protocol setup: distance between states}

We first explain how estimating the Gram matrix leads to a quantitative bound on the distance between external states. Assume first that the input state is pure. As discussed in \Cref{sec:partialdistTheory}, equation \eqref{eq:omegaext-gram} is invariant under rephasing $G \mapsto D^{\dagger}GD$, where $D=\operatorname{diag}(e^{i\varphi_1},\ldots,e^{i\varphi_n})$ corresponds to $\ket{\psi_j} \mapsto e^{i\varphi_j}\ket{\psi_j}$. We denote the resulting equivalence class by
\begin{equation}
    \equivclass(\psi) := [G] = \{D^{\dagger}GD \ \vert \ \text{$D$ diagonal unitary} \}. 
\end{equation}
The Gram matrix $G$, or more precisely the equivalence class $\equivclass(\psi)$, must be recovered experimentally and thus can be known only approximately. Suppose our experimental estimate is $G'$, which corresponds to some pure product state $\psi'$. The following proposition shows that the error between $G$ and $G'$ controls the error between corresponding external states $\Omegaext(\psi)$ and $\Omegaext(\psi')$. Thus, to perform internal state tomography it suffices to recover the Gram matrix up to desired error. Formally for any $p\geq 1$ let us define the distance between equivalence classes:
\begin{equation}\label{eq:def-distance-equivalence}
    d_p(\equivclass(\psi),\equivclass(\psi')) := \inf_{\substack{G \in \equivclass(\psi) \\ G' \in \equivclass(\psi')}}\norm{G - G'}_{p}
\end{equation}
where $\norm{\cdot}_{p}$ is the usual matrix Schatten norm (for $p=2$ this is simply the Frobenius norm).

\begin{proposition}\label{prop:tvd-continuity}
    Suppose that $\psi, \psi' \in \hint^{\otimes n}$ are internal pure product states. Then:
    \begin{equation}
        \norm{\Omegaext(\psi) - \Omegaext(\psi')}_1 \leq d_1(\equivclass(\psi),\equivclass(\psi'))
    \end{equation}
     In particular using the inequality between the $1$ and $2$-norm we have:
    \begin{equation}\label{eq:prop-gg2}
        \norm{\Omegaext(\psi) - \Omegaext(\psi')}_1 \leq \sqrt{n}\cdot d_2(\equivclass(\psi),\equivclass(\psi'))
    \end{equation}
\end{proposition}
Due to the infimum in \eqref{eq:def-distance-equivalence} any choice of $G \in \equivclass(\psi), G' \in \equivclass(\psi')$ gives a valid upper bound on the distance. We defer the proof of \Cref{prop:tvd-continuity} to \Cref{sec:proof-of-tvd}.

An immediate physical consequence of \Cref{prop:tvd-continuity} is that the observable statistics of the two states will also be closely matched. By the standard operational interpretation of the trace distance, the total variation distance (TVD) between the probability distributions of any measurement outcomes is bounded by the trace distance. Applying this to a standard optical setup -- where internal states cannot be resolved -- yields the following bound.

\begin{corollary}\label{cor: total variation distance}
Let $\psi, \psi'$ be internal pure product states. Let $p_{\mathcal{M}}(\psi)$ and $p_{\mathcal{M}}(\psi')$ denote the probability distributions of outcomes obtained by performing any arbitrary generalized measurement (POVM) $\mathcal{M}$ on the external states $\Omega_{ext}(\psi)$ and $\Omega_{ext}(\psi')$, respectively. Then we have:
\begin{align}
    \TVD(p_{\mathcal{M}}(\psi), p_{\mathcal{M}}(\psi')) &\leq \frac{1}{2}d_1(\equivclass(\psi),\equivclass(\psi'))\\ &\leq \frac{\sqrt{n}}{2}d_2(\equivclass(\psi),\equivclass(\psi'))
\end{align}
\end{corollary}

 
If the input state is pure, by \Cref{prop:tvd-continuity} it suffices to reconstruct its (equivalence class of) Gram matrix. To handle mixed input states, let assume that:       
\begin{equation}
    \rho = \bigotimes_{i=1}^{n}\rho_i
\end{equation} is a mixed internal product state such that for each $i$ the purity satisfies:
\begin{equation}\label{eq:purity-parameter}
    \tr(\rho_i^2) > 1 - \varepsilon
\end{equation} for some small $\varepsilon$. By \cite{oszmaniecMeasuringRelationalInformation2021}, for each $i$ there exists a pure state $\psi_i$ such that $\norm{\rho_i - \psi_i}_{1} < 2\varepsilon$ ($\psi_i$ can be taken as the projector onto the maximal eigenvector of $\rho_i$). Let $\psi := \bigotimes_{i=1}^{n}\psi_i$ and let $G$ be the Gram matrix corresponding to $\psi$. By tensorization:
\begin{equation}\label{eq:trace-inequality}
    \norm{\rho - \psi}_{1} \leq 2n\varepsilon
\end{equation}
which also implies:
\begin{equation}\label{eq:trace-inequality-external}
    \norm{\Omegaext(\rho) -\Omegaext(\psi)}_{1} < 2n\varepsilon
\end{equation}
We aim to reconstruct $\Omegaext(\rho)$ by approximately reconstructing $\Omegaext(\psi)$, using Bargmann invariants of $\rho$ and the protocol as described for pure states. By \cite{oszmaniecMeasuringRelationalInformation2021} (Corollary 1) for any Bargmann invariant of order $m$ we have:
\begin{equation}\label{eq:bargmann-diff}
    \vert \tr(\rho_{i_{1}}\dots \rho_{i_{m}}) - \tr (\psi_{i_{1}}\dots \psi_{i_{m}}) \vert \leq 2m \varepsilon
\end{equation}
Thus, we can treat measurements of Bargmann invariants of $\rho$ as approximate measurements of Bargmann invariants of $\psi$. In other words, the Bargmann invariants of $\rho$ do not correspond directly to any Gram matrix (since the state is mixed, there is no Gram matrix to speak of) -- however, the pure state $\psi$ can be used as a proxy for $\rho$ and its Gram matrix can be (approximately) recovered from the Bargmann invariants of $\rho$. The total error incurred for the Gram matrix entries will include the sampling error depending on $N$, the number of experimental samples, and systematic error depending on the purity parameter $\varepsilon$. The latter persists even as $N \to \infty$ since the protocol approximates the target state $\rho$ using its closest pure state $\psi$. 

We note that the task of Gram matrix tomography requires significantly more resources than the simpler task of characterizing indistinguishability. It is thus reasonable to consider if reconstructing the Gram matrix can also bound or estimate the amount of indistinguishability that the set of $n$ photons carry. In Section  \ref{Sec: Certification} we show that the pure state assumption is enough to fully accomplish this task up to small error.

\subsection{Estimation procedure for entries of \texorpdfstring{$G$}{G}}\label{sec:protocol-main}

We now derive the protocol underlying \Cref{th:protocol} and \Cref{th:protocol2}. The protocol assumes a mixed product input state $\rho$, with pure input states as a special case. The goal is to reconstruct from experiment the matrix entries of $G$, the Gram matrix of the pure product state $\psi$ close to $\rho$ that satisfies \eqref{eq:trace-inequality-external}. For a pure state, it is known \cite{oszmaniecMeasuringRelationalInformation2021} how to reconstruct its Gram matrix given access to $O(n^2)$ Bargmann invariants. This procedure is subject to two sources of error. 

First, the internal states featuring in the experiment are not the pure states $\psi_{i}$, rather, what is experimentally accessible are the Bargmann invariants of the mixed states $\rho_s$ which satisfy \eqref{eq:bargmann-diff}. This results in a \textit{systematic error} in the protocol that depends on the purity parameter $\varepsilon$ and persists in the limit $N\to\infty$. The second source of error is that the Bargmann invariants $\tr(\rho_{i_1}\dots\rho_{i_m})$ are themselves known only approximately, since they need to be estimated experimentally from measurements. This incurs a \textit{sampling error} dependent on the number of samples. These two types of error together add up to the total error of the protocol.

Assume that for each tuple $(i_1,\dots,i_m)$ we have an estimator $\widehat{B}_{i_1\dots i_m}$ of the mixed Bargmann invariant $\tr(\rho_{i_{1}}\dots \rho_{i_{m}})$. We assume for simplicity the following experimental setting: constrained to a total budget of $N$ experimental samples, we perform measurements of $k$ different interferometers, using $\frac{N}{k}$ samples for each, where in each experimental run we measure one or more Bargmann invariants. For measurement methods described in \Cref{sec:relational-info} we then have $k \sim n^2$ for the Fourier method (each required Bargmann invariant needs a separate interferometer) while $k=1$ for the one-shot method (all required Bargmann invariants are estimated at once).

We assume the estimators are unbiased, i.e.:
\begin{equation}
    \E (\widehat{B}_{i_1\dots i_m}) = \tr(\rho_{i_{1}}\dots \rho_{i_{m}})
\end{equation}
with the expectation taken over the experimental sample randomness. We assume that all estimator variances are bounded by a parameter $\sigma^2$, so that:
\begin{equation}\label{eq:bargmann-variance}
    \variance(\widehat{B}_{i_1 \dots i_m}) \leq \frac{\sigma^2}{\frac{N}{k}}
\end{equation}
where $N$ is the number of experimental samples.

 We shall treat $\widehat{B}_{i_1\dots i_m}$ as a \textit{biased} estimator of $B_{i_1 \dots i_{m}} := \tr(\psi_{i_{1}} \dots \psi_{i_{m}})$, with bias :
\begin{equation}\label{eq:bargmann-bias}
    \bias(\widehat{B}_{i_1\dots i_m}) := \tr(\rho_{i_{1}}\dots \rho_{i_{m}}) - \tr(\psi_{i_{1}} \dots \psi_{i_{m}})
\end{equation}
Note that by \eqref{eq:bargmann-diff} the bias is bounded by 
\begin{equation}\label{eq:bias-bound}
\abs{  \bias(\widehat{B}_{i_1\dots i_m})} \leq 2m \varepsilon.    
\end{equation}
The mean square error for a single invariant is then:
\begin{align}\label{eq:bargmann-MSE}
    \E\abs{ B_{i_1\dots i_m} - \widehat{B}_{i_1\dots i_m} }^2 &= \bias(\widehat{B}_{i_1\dots i_m})^2 + \variance(\widehat{B}_{i_1\dots i_m})\\ &\leq 4m^2 \varepsilon^2 + \frac{k\sigma^2}{N}
\end{align}

Let $G$ denote the target Gram matrix of $\psi = \bigotimes_{i=1}^{n}\ket{\psi_i}\bra{\psi_i}$. To provide intuition for the protocol we first recall from \cite{oszmaniecMeasuringRelationalInformation2021} the setup in the ideal case where we simply want to reconstruct $G$ from perfect knowledge of Bargmann invariants $B_{i_1 \dots i_m}$. Let $\Gamma$ be the graph on the vertex set $\{1,\dots,n\}$ with an edge $(i,j)$ present if $G_{ij} \neq 0$. We assume that $\Gamma$ is connected, otherwise we can perform the procedure on each connected component separately. Fix a spanning tree $T$ of the graph $\Gamma$. By using gauge equivalence \cite{oszmaniecMeasuringRelationalInformation2021}, we can assume that for each edge $(i,j) \in T$ the entry $G_{ij}$ is real and nonnegative. This is equivalent to choosing a particular representative $G$ of the equivalence class $\equivclass(\psi)$. It follows that for $(i,j) \in T$:
 \begin{equation}\label{eq:gram-ideal1}
    G_{ij} =  \abs{\braket{\psi_i}{ \psi_j}} =  \sqrt{B_{ij}}.
\end{equation}
For $(i,j) \notin T$ let $\pi(i,j)=(i=u_1, u_2,\dots, u_{k-1},u_k=j)$ denote the unique path in the tree $T$ joining $i$ and $j$. The corresponding Bargmann invariant is  $B_{\pi(i,j)} = \tr(\psi_{u_{1}}\dots\psi_{u_{k}})$. We then easily verify that: 
\begin{equation}\label{eq:b-notin-T}
    G_{ij} = \frac{B_{\pi(j,i)}}{\prod_{l=1}^{k-1} \sqrt{B_{u_{l},u_{l+1}}}}
\end{equation}
Let $\varphi_{ij}$ denote the complex phase of $B_{\pi(i,j)}$, so that:
\begin{equation}
    e^{i\varphi_{ij}} = \frac{B_{\pi(i,j)}}{\abs{B_{\pi(i,j)}}}
\end{equation}
We then easily verify that:
\begin{equation}\label{eq:gram-ideal2}
    G_{ij} = e^{i\varphi_{ji}} \sqrt{B_{ji}}
\end{equation}

Therefore, knowledge of all the invariants $B_{ij}$ and all the invariants $B_{\pi(i,j)}$ corresponding to paths in the tree $T$ (or equivalently all the phases $\varphi_{ij}$), allows us to fully recover the Gram matrix $G$, in the choice of gauge given by this particular spanning tree.

To pass from the ideal case to the non-ideal case, we need to replace quantities appearing in equations \eqref{eq:gram-ideal1}, \eqref{eq:gram-ideal2} with their appropriate estimators. In this way we will obtain, entry by entry, a (biased) estimator $\widehat{G}$ of the target Gram matrix $G$. We will measure the total error of this estimator using the mean squared error (MSE): 
\begin{equation}
    \E\norm{G-\widehat{G}}^2_{2} = \sum_{i,j}\E\abs{G_{ij} - \widehat{G}_{ij}}^2 
\end{equation}
Note that by \Cref{prop:tvd-continuity} we could also use the $1$-norm $\E\norm{G - \widehat{G}}_{1}$, which gives a tighter bound on the distance between external states (by a factor of $\sqrt{n}$) but is more difficult to control analytically.

We now derive the full protocol and analyze its MSE error. First, for each pair $(i,j)$ we measure the estimator $\widehat{B}_{ij}$. Since very small values of $\widehat{B}_{ij}$ will be numerically unstable when computing square roots, we fix a threshold $\eta > 0$ and set $\widehat{G}_{ij}=0$ for all $(i,j)$ such that $\abs{\widehat{B}_{ij}} < \eta$. We remove these edges from the graph. The remaining pairs $(i,j)$ form a graph $\Gamma$, which we again assume to be connected.  We fix a spanning tree $T$ of $\Gamma$, chosen in a way that will minimize the incurred error, see \Cref{sec:tree} for the detailed procedure. 

We use the following auxilliary estimators. To estimate the square root of a Bargmann invariant we plug:
\begin{equation}
    \widehat{\sqrt{B_{ij}}} := \sqrt{\widehat{B}_{ij}}
\end{equation}
and for the phase part of $B_{\pi(i,j)}$ we plug:
\begin{equation}\label{eq:estimator-phase}
    \widehat{e^{i\varphi_{i,j}}} := \frac{\widehat{B}_{\pi(i,j)}}{\abs{\widehat{B}_{\pi(i,j)}}}
\end{equation}
We now define the final estimator $\widehat{G}$ for the Gram matrix. For diagonal terms we put $\widehat{G}_{ii}:=1$. For an edge $(i,j)$ we define the estimator $\widehat{G}_{ij}$ as follows:
\begin{enumerate}[(a)]
    \item if $(i,j) \in T$, then:
    \begin{equation}\label{eq:gram1}
        \widehat{G}_{ij} := \widehat{\sqrt{B_{ij}}}
    \end{equation}
    \item if $(i,j) \notin T$, then:
    \begin{equation}\label{eq:gram2}
    \widehat{G}_{ij} :=  \widehat{e^{i\varphi_{i,j}}} \cdot \widehat{\sqrt{B_{ij}}}
\end{equation}
\end{enumerate}
Note that the maximal order of Bargmann invariants that we need to measure is equal to diameter of the spanning tree plus one. 

The resulting estimator $\widehat{G}$ may not be a valid Gram matrix. To mitigate this we project $\widehat{G}$ onto the convex set of correlation matrices (positive semidefinite with ones on the diagonal), obtaining a matrix $p(\widehat{G})$. Since projections onto convex sets do not increase the Euclidean distance and $p(G)=G$ (since $G$ itself is a correlation matrix), we have: 
\begin{equation}
    \norm{p(\widehat{G}) - G}_2 = \norm{p(\widehat{G}) - p(G)}_2 \leq \norm{\widehat{G} - G}_2
\end{equation}
so the projection step does not increase the total error. The projection step can be efficiently implemented e.g. using Higham's projection algorithm \cite{nigham}.

The whole protocol can be summarized as follows.





\begin{mdframed}[
  linewidth=1pt,
  linecolor=black,
  backgroundcolor=gray!10,
  innertopmargin=6pt,
  innerbottommargin=8pt
]
\textbf{Gram matrix tomography protocol}
\begin{enumerate}
    \item For each pair $(i,j)$ measure the Bargmann invariant $B_{ij}$ using an estimator $\widehat{B}_{ij}$.
    \item Fix a threshold $\eta$ and remove edges with $\vert \widehat{B}_{ij} \vert < \eta$. 
    \item Find the spanning tree $T$ on the remaining graph using the algorithm from \Cref{sec:tree}.
    \item Based on the spanning tree measure the remaining Bargmann invariants $B_{\pi(i,j)}$ using an estimator $\widehat{B}_{\pi(i,j)}$. Estimate the phases using \eqref{eq:estimator-phase}.
    \item Estimate the Gram matrix entries $\widehat{G}_{ij}$ using \eqref{eq:gram1} and \eqref{eq:gram2}.
    \item Project the resulting matrix $\widehat{G}$ onto the set of correlation matrices.
\end{enumerate}
\end{mdframed}

For the purpose of analysis it is required that the choice of the spanning tree $T$ is statistically independent of the randomness of the Bargmann invariant estimators measured in step 5. This can be easily achieved e.g. by splitting the sample in half, using estimators from the first half to produce the optimal spanning tree $T$, discarding them and using the second half of the sample for the remaining steps of the protocol.

\begin{figure*}[tbhp!]
  \centering
  \begin{tikzpicture}


  \node[stepBox] (B1) at (0, 0) {};
  \node[stepNum] at ($(B1.north west)+(-4pt, 0pt)$) {1};
  \node[captionText, anchor=north] at ($(B1.north)+(0, -7pt)$) {Measure pairwise invariants};

  \coordinate (v1_1) at ($(B1.center)+(-1.15,  0.75)$);
  \coordinate (v2_1) at ($(B1.center)+( 1.15,  0.75)$);
  \coordinate (v3_1) at ($(B1.center)+(-1.15, -0.95)$);
  \coordinate (v4_1) at ($(B1.center)+( 1.15, -0.95)$);

  \draw[baseEdge] (v1_1) -- (v2_1);
  \draw[baseEdge] (v1_1) -- (v3_1);
  \draw[baseEdge] (v2_1) -- (v4_1);
  \draw[baseEdge] (v3_1) -- (v4_1);
  \draw[baseEdge] (v1_1) -- (v4_1);
  \draw[baseEdge] (v2_1) -- (v3_1);

  \pic at ($(v1_1)!0.5!(v2_1)$) {meter};
  \pic at ($(v3_1)!0.5!(v4_1)$) {meter};
  \pic at ($(v1_1)!0.5!(v3_1)$) {meter};
  \pic at ($(v2_1)!0.5!(v4_1)$) {meter};
  \pic at ($(v1_1)!0.32!(v4_1)$) {meter};
  \pic at ($(v2_1)!0.32!(v3_1)$) {meter};

  \node[gnode] at (v1_1) {};
  \node[gnode] at (v2_1) {};
  \node[gnode] at (v3_1) {};
  \node[gnode] at (v4_1) {};

  \node[stepBox] (B2) at (5.5, 0) {};
  \node[stepNum] at ($(B2.north west)+(-4pt, 0pt)$) {2};
  \node[captionText, anchor=north] at ($(B2.north)+(0, -7pt)$) {Remove small edges ($|\hat{B}_{ij}| < \eta$)};

  \coordinate (v1_2) at ($(B2.center)+(-1.15,  0.75)$);
  \coordinate (v2_2) at ($(B2.center)+( 1.15,  0.75)$);
  \coordinate (v3_2) at ($(B2.center)+(-1.15, -0.95)$);
  \coordinate (v4_2) at ($(B2.center)+( 1.15, -0.95)$);

  \draw[baseEdge] (v1_2) -- (v2_2);
  \draw[baseEdge] (v1_2) -- (v3_2);
  \draw[baseEdge] (v2_2) -- (v4_2);
  \draw[baseEdge] (v1_2) -- (v4_2);

  \node[gnode] at (v1_2) {};
  \node[gnode] at (v2_2) {};
  \node[gnode] at (v3_2) {};
  \node[gnode] at (v4_2) {};

  \node[stepBox] (B3) at (11.0, 0) {};
  \node[stepNum] at ($(B3.north west)+(-4pt, 0pt)$) {3};
  \node[captionText, anchor=north] at ($(B3.north)+(0, -7pt)$) {Choose a spanning tree $T$};

  \coordinate (v1_3) at ($(B3.center)+(-1.15,  0.75)$);
  \coordinate (v2_3) at ($(B3.center)+( 1.15,  0.75)$);
  \coordinate (v3_3) at ($(B3.center)+(-1.15, -0.95)$);
  \coordinate (v4_3) at ($(B3.center)+( 1.15, -0.95)$);

  \draw[baseEdge] (v1_3) -- (v2_3);
  \draw[treeEdge] (v1_3) -- (v3_3);
  \draw[treeEdge] (v1_3) -- (v4_3);
  \draw[treeEdge] (v2_3) -- (v4_3);

  \node[gnode] at (v1_3) {};
  \node[gnode] at (v2_3) {};
  \node[gnode] at (v3_3) {};
  \node[gnode] at (v4_3) {};


  \node[stepBox] (B4) at (11.0, -5.2) {};
  \node[stepNum] at ($(B4.north west)+(-4pt, 0pt)$) {4};
  \node[captionText, anchor=north] at ($(B4.north)+(0, -7pt)$) {Measure cycle phases $\hat{\varphi}_{ij}$};

  \coordinate (v1_4) at ($(B4.center)+(-1.15,  0.38)$);
  \coordinate (v2_4) at ($(B4.center)+( 1.15,  0.38)$);
  \coordinate (v3_4) at ($(B4.center)+(-1.15, -0.74)$);
  \coordinate (v4_4) at ($(B4.center)+( 1.15, -0.74)$);

  \draw[baseEdge] (v1_4) -- (v2_4);
  \draw[treeEdge] (v1_4) -- (v3_4);
  \draw[treeEdge] (v1_4) -- (v4_4);
  \draw[treeEdge] (v2_4) -- (v4_4);

  \draw[line width=1.5pt, draw=crimsonRed, dashed, -{Latex[scale=0.9]}]
    (v1_4) to[out=38, in=142, looseness=0.95] coordinate[midway] (arcMid) (v2_4);

  \pic at (arcMid) {meter};


  \node[font=\small] at ($(B4.center)+(0, -1.38)$) {
    $e^{i\hat{\varphi}_{ij}} = \dfrac{\hat{B}_{\pi(i,j)}}{|\hat{B}_{\pi(i,j)}|}$
  };

  \node[gnode] at (v1_4) {};
  \node[gnode] at (v2_4) {};
  \node[gnode] at (v3_4) {};
  \node[gnode] at (v4_4) {};

  \node[stepBox] (B5) at (5.5, -5.2) {};
  \node[stepNum] at ($(B5.north west)+(-4pt, 0pt)$) {5};
  \node[captionText, anchor=north] at ($(B5.north)+(0, -7pt)$) {Construct the matrix $\widehat{G}$};

  \node at ($(B5.center)+(0, -0.22)$) {
    $\widehat{G} = \begin{pmatrix}
      1 & \hat{G}_{12} & \cdots & \hat{G}_{1n} \\
      \hat{G}_{21} & 1 & \cdots & \hat{G}_{2n} \\
      \vdots & \vdots & \ddots & \vdots \\
      \hat{G}_{n1} & \hat{G}_{n2} & \cdots & 1
    \end{pmatrix}$
  };

  \node[stepBox] (B6) at (0, -5.2) {};
  \node[stepNum] at ($(B6.north west)+(-4pt, 0pt)$) {6};
  \node[captionText, anchor=north] at ($(B6.north)+(0, -7pt)$) {Project onto Gram matrices $p(\widehat{G})$};

  \node[font=\Large] at ($(B6.center)+(0, -0.2)$) {
    $\begin{aligned}
      &p(\widehat{G}) \succeq 0 \\[8pt]
      &p(\widehat{G})_{ii} = 1
    \end{aligned}$
  };

  \draw[flowArrow] (B1.east) -- (B2.west);
  \draw[flowArrow] (B2.east) -- (B3.west);
  \draw[flowArrow] (B3.south) -- (B4.north);
  \draw[flowArrow] (B4.west) -- (B5.east);
  \draw[flowArrow] (B5.west) -- (B6.east);

  \end{tikzpicture}
  \caption{\textbf{Reconstruction pipeline for Gram-matrix tomography.} \textbf{Steps 1--2 (Graph initialization and thresholding):} Pairwise Bargmann invariants $B_{ij} = |\langle \psi_i | \psi_j \rangle|^2$ are measured to extract the magnitudes $|G_{ij}|$ of all off-diagonal entries. Edges falling below an operational cutoff threshold $\eta > 0$ are removed to prevent noise amplification and numerical instabilities when computing square roots and phase normalizations. \textbf{Step 3 (Gauge fixing):} Since rephasing individual state vectors transforms $G$ by an unobservable diagonal unitary $D^\dagger G D$, a spanning tree $T$ is selected on the remaining overlap graph to fix this gauge convention by setting all tree-edge overlaps to be real and positive.  \textbf{Step 4 (Phase recovery):} Each non-tree edge closes a loop with the unique path $\pi(i,j)$ in $T$. Measuring the corresponding higher-order cyclic Bargmann invariant $B_{\pi(i,j)}$ along this path directly yields the relative phase factor $e^{i\varphi_{ij}} = B_{\pi(i,j)} / |B_{\pi(i,j)}|$. \textbf{Steps 5--6 (Assembly and convex projection):} Combining the measured moduli and phases yields an intermediate estimator $\hat{G}$. To correct for finite-sample violations of positive semidefiniteness, $\hat{G}$ is projected onto the convex set of Gram matrices ($p(\hat{G}) \ge 0$, $p(\hat{G})_{ii} = 1$), which enforces physical validity without increasing the Frobenius reconstruction error.}
\label{fig:reconstruction_protocol}
  \label{fig: Main idea}
\end{figure*}
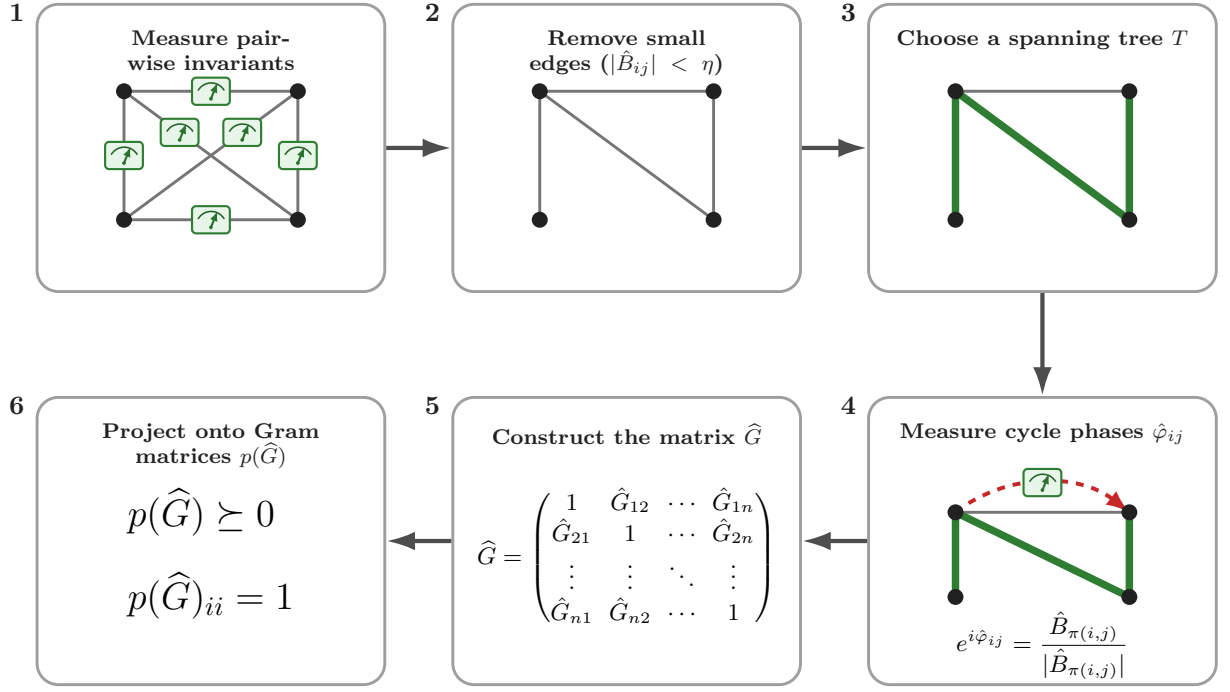

\begin{proposition}\label{th:main-MSE-simple}
    Consider a mixed state $\rho$ with purity parameter $\varepsilon$ as in \eqref{eq:purity-parameter} and estimators of Bargmann invariants using in total $N$ experimental samples. For threshold value $\eta=0$, the above protocol gives an estimator $\widehat{G}$ for the Gram matrix $G$ with mean square error bounded by:
    \begin{equation}\label{eq:main-theorem-mse-bound}
        \E\norm{G - \widehat{G}}_{2}^{2} \leq \frac{A(n)}{N}+B(n)\varepsilon^2 
    \end{equation}
    where $A(n),B(n)$ are explicit functions depending on the spanning tree $T$, estimators used and the magnitudes of Bargmann invariants.
\end{proposition}

The proof of \Cref{th:main-MSE-simple} is technical and we defer it to \Cref{sec:estimators}, where we also give a more detailed expression and bounds for the functions $A$ and $B$. In particular, if some of the Bargmann invariants are very small, it is useful to fix a nonzero threshold $\eta >0$ and we describe its influence on the total error in \Cref{sec:estimators}. 


\begin{proof}[Proof of \Cref{th:protocol} and \Cref{th:protocol2}]
Let $\psi$ be the pure state from \eqref{eq:trace-inequality} and $\psi'$ be the state whose Gram matrix is reconstructed by the protocol. From \eqref{eq:prop-gg2} we get:
\begin{equation}
    \E \norm{\Omegaext(\psi) -\Omegaext(\psi')}^2_{1} \leq n\cdot \E\norm{G-\widehat{G}}^2_2 \leq  n\left( \frac{A}{N} + B \varepsilon^2 \right)
\end{equation}
By \eqref{eq:trace-inequality-external} and the triangle inequality we immediately get:
\begin{gather}\label{eq:pre-markov}
    \E \norm{\Omegaext(\rho) -\Omegaext(\psi')}^2_{1} \leq 8n^2\varepsilon^2 +  n\left(  B \varepsilon^2 + \frac{A}{N}\right)
\end{gather}

Using Markov inequality now gives:
\begin{align}
    \mathbb{P}\left( \norm{\Omegaext(\rho) -\Omegaext(\psi')}_{1} \geq \delta \right) &\leq \frac{\E\norm{\Omegaext(\rho) -\Omegaext(\psi')}^2_{1}}{\delta^2}\\ &\leq \frac{B' \varepsilon^2 + \frac{A'}{N}}{\delta^2}
\end{align}
as desired, where $A', B'$ are deduced from \eqref{eq:pre-markov}.
        \end{proof}

\section{Accessing relational information}\label{sec:relational-info}

The reconstruction protocol introduced in Sec.~\ref{sec: Tomography} relies on the experimental accessibility of the relevant Bargmann invariants. While  characterizing an arbitrary mixed state generally involves invariants  across all cycle types, reconstructing pure or nearly pure states requires only a subset of $O(n^2)$ invariants: pairwise overlaps to fix the moduli of  the Gram matrix and cycle invariants along paths in the spanning tree to fix the relative phases. The order of these cycle invariants is determined by the tree topology and can in principle reach $n$. Long cycles, however, are only required when many pairwise overlaps are small, a situation that rarely occurs for realistic sources, where photons typically exhibit substantial mutual overlap. In particular, whenever one photon has a non-negligible overlap with all the others, a star-shaped spanning tree can be chosen and all required invariants are of third order.

In this section, we present two strategies to measure these invariants using passive linear optics and number-resolved detection. The first approach (Fourier method) uses a Fourier interferometer to individually target the needed invariants, while the second approach (one-shot method) extracts relational information simultaneously from the particle-number correlators of a single generic interferometer through linear inversion.

\paragraph{Fourier method.-} As discussed above, a Bargmann invariant can be expressed as the expectation value of the cyclic shift operator,
\begin{equation}
    \Tr{\U{C}_{k}\bigotimes_{i=1}^{k}\rho_{i}}=\Tr{\rho_{1}...\rho_{k}}
\end{equation}
where $\U{C}_{k}|\psi_{1},\psi_{2},\dots,\psi_{k}\rangle=|\psi_{2},\dots,\psi_{k},\psi_{1}\rangle$. As shown in \cite{novoNativeLinearopticalProtocol2026}, in linear optics, such expectation value can be evaluated by diagonalizing the cyclic shift operator $\U{C}_{k}=\U{F}_{k}\U{D}_{k}\U{F}_{k}^{\dagger}$, where $\U{F}_{k}$ is the $k$-dimensional Fourier interferometer, and $\hat{D}$ is an operator diagonal in the Fock basis, which can be expressed as 
\begin{equation}
    \U{D}_{k}=\prod_{j=1}^{k}e^{\frac{2\pi i}{k}(j-1)\hat{n}_{j}} .
\end{equation}
This implies that the task of measuring an arbitrary Bargmann invariant can be reduced to the sampling problem, i.e. let $\Omegaext(\rho)$ evolve in the interferometer $\U{F}_{k}\oplus \U{I}_{n-k}$ and perform particles number detection on the output, to then evaluate $\langle \U{D}_{k}\rangle$ by postprocessing the data. Since $\lVert\U{D}_{k}\rVert_{\infty}=1$, estimating any Bargmann invariant up to precision $\epsilon$ and probability $1-\delta$ can be done using $\mathcal{O}\left(\epsilon^{-2}\ln\delta^{-1}\right)$ samples. Notice that for Bargmann of order $2$, the Fourier matrix reduces to a balanced beam splitter and the protocol coincides with a standard Hong-Ou-Mandel (HOM) test. The complete procedure is summarized below.

\begin{algorithm}[H]
\caption{Fourier method for Gram matrix tomography}
\label{Prot: Fourier}
\begin{algorithmic}
\State \textbf{Input:} external state \(\Omega_{\mathrm{ext}}(\rho)\)
\For{\( 1\leq i<j\leq n\)}
    \State Use the HOM interference to estimate the visibility to estimate $|G_{ij}|^{2}$
\EndFor
\State Use the value of $|G_{ij}|^{2}$ to find the Spanning tree $T$
\For{\( \text{each } (i,j)\not\in T\)}
    \State Use the Fourier interferometer to estimate $B_{\pi(i,j)}=\Tr{\rho_{i}\rho_{k_{1}}\dots \rho_{k_{l}}\rho_{j}}$ and obtain $e^{i\phi_{i,j}}$.
\EndFor
\State Construct $\hat{G}$
\State Project over the space of Gram matrix $\hat{G}\gets p(\hat{G})$
\State \textbf{Return} $\hat{G}$
\end{algorithmic}
\label{Protocol: Fourier}
\end{algorithm}

Notice that this protocol has an intrinsic overhead in the sample complexity due to the fact that each of the $\mathcal{O}(n^{2})$ quantities must be measured separately, and thus $\mathcal{O}(n^{2})$ experiments must be run to obtain a good estimate of the desired Gram matrix. 


\paragraph{One-Shot method.-} The one-shot method follows a different route, exploiting the fact that partial distinguishability affects the output statistics of multiphoton scattering experiments such as Boson Sampling, in which single photons evolve through a linear interferometer and particle numbers are measured at the output~\cite{aaronson2010computationalcomplexitylinearoptics}. It can be shown that, in the interferometric scenario we consider, the probability of detecting an outcome $\boldsymbol{k}=(k_{1},...,k_{m})$ when the unitary evolution is described by the matrix $U$ can be computed as 
\begin{equation}
    p_{\boldsymbol{k}}(U|G)=\frac{1}{\boldsymbol{k}!}\sum_{\pi \in S_{n}}\perm\left(M^{(\boldsymbol{k})}\odot\overline{M}_{\pi}^{(\boldsymbol{k})}\right)\prod_{i=1}^{n}G_{i,\pi_{i}}
\end{equation}
where $\boldsymbol{k}!=\prod_{i}k_{i}!$, $\perm(A)$ is the permanent of the matrix $A$, $M^{(\boldsymbol{k})}$ is the relevant scattering matrix associated with the outcome $\boldsymbol{k}$, $\odot$ is the element-wise product and $\overline{A}_{\sigma}$ is the matrix conjugate of $A$ whose columns have been permuted according to the permutation $\pi$ \cite{tichySamplingPartiallyDistinguishable2015}. Extracting a specific invariant from these probabilities, i.e. $\Tr\left[\rho_{1}\cdots\rho_{k}\right]=G_{12}G_{23}\cdots G_{k1}$, would require estimating typically very small outcome probabilities to high precision and inverting a linear system whose number of terms grows factorially with $n$. This is feasible for few photons ($n=3,4$), but quickly becomes impractical as $n$ increases. An alternatively scalable observable can be found in the number particles correlators. As shown in Refs.~\cite{walschaersStatisticalBenchmarkBosonSampling2016, rodariSemideviceIndependentCharacterization2024, anguitaExperimentalValidationBoson2025, robbio2026complementaritybosonicfermionicmanybody}, it is possible to directly observe deviation from the behavior of indistinguishable particles by measuring correlators of type $\langle \hat{n}_{i}\hat{n}_{j}\rangle$. We note that such correlators may be written in the form
\begin{align}\label{eq:correlators2}
    \langle \hat{n}_{i}\hat{n}_{j}\rangle &= \delta_{ij} \sum_{k=1}^n |U_{ik}|^2 + \sum_{k \neq l}^{n}  |U_{ik}|^2 |U_{jl}|^2 + \\
    &+ \sum_{k \neq l}^{n}|G_{kl}|^2 U_{ik}^* U_{il} U_{jl}^* U_{jk}  \notag
\end{align}  
The second moments are thus affine in the squared overlaps $|G_{kl}|^{2}$, i.e., in the second-order Bargmann invariants. Grouping the equivalent terms $(k,l)$ and $(l,k)$, they define a linear system $A_{2}x_{2}=b_{2}$, where $x_{2}$ collects the unknown squared overlaps, $A_{2}$ depends only on $U$, and $b_{2}$ combines the measured correlators with $U$-dependent terms:
\begin{align}
    (b_{2})_{(ij)}=\langle \hat{n}_{i}\hat{n}_{j}\rangle+ \sum_{k \neq l} |U_{ik}|^2 |U_{jl}|^2 - \delta_{ij} \sum_{k=1}^n |U_{ik}|^2 \quad ,\label{eq:Axb1} \\
    (x_{2})_{(ij)}=|G_{ij}|^{2} \quad , \quad (A_{2})_{(ij),(kl)}=U_{ik}^* U_{il} U_{jl}^* U_{jk} \quad \label{eq:Axb2}.
\end{align}
Once $\langle \hat{n}_{i}\hat{n}_{j}\rangle$ has been estimated for all $i,j$ from the collected samples, inverting this system (e.g., by least squares) yields estimators of $|G_{kl}|^{2}$. Similarly, the third-order correlators $\langle \hat{n}_{i}\hat{n}_{j}\hat{n}_{k}\rangle$ depend linearly on the second-order invariants $|G_{ab}|^{2}$ and on the third-order invariants $G_{ab}G_{bc}G_{ca}$. Since the former are already known from the previous step, the latter are the only unknowns and can again be obtained by solving a linear system $A_{3}x_{3}=b_{3}$. A detailed derivation of the correlators is given in Appendix~\ref{sec: Correlator derivation}. In principle, the procedure extends to higher orders, as $r$-th order correlators depend linearly on invariants up to order $r$. However, the number of unknowns grows as $\mathcal{O}(n^{r})$, which makes the approach impractical beyond low orders. It nonetheless remains well suited to the regime discussed above, in which third-order invariants suffice for the reconstruction. Additionally, in Appendix \ref{sec: size reduction} we discuss how to reduce the size of the system to only the $\mathcal{O}(n^2)$ invariant needed without affecting the numerical stability of the linear system. The complete procedure is summarized below.

\begin{algorithm}[H]
\caption{One-Shot method for Gram matrix tomography}
\label{Prot: Oneshot}
\begin{algorithmic}
\State \textbf{Input:} external state \(\Omega_{\mathrm{ext}}(\rho)\), random interferometer $\U{U}$
\State Sample in the Fock basis from $\U{U}\Omegaext(\rho)\U{U}^{\dagger}$ and collect samples $\{\boldsymbol{s}^{(1)},...,\boldsymbol{s}^{(M)}\}$.
\State Build the matrix $A_{2},A_{3}$ and the vectors $b_{2},b_{3}$ from $\{\boldsymbol{s}^{(1)},...,\boldsymbol{s}^{(M)}\}$ and $U$.
\State Solve the linear $x_{2}=A_{2}b_{2}$ to obtain $\{|\hat{G}_{ij}|^{2}\}$.
\State Use the value of $|\hat{G}_{ij}|^{2}$ to find the Spanning tree $T$.
\State Solve the linear $x_{3}=A_{3}b_{3}$ to obtain $\{\hat{G}_{ij}\hat{G}_{jk}\hat{G}_{ki}\}$.
\State Construct $\hat{G}$
\State Project over the space of Gram matrix $\hat{G}\gets p(\hat{G})$
\State \textbf{Return} $\hat{G}$
\end{algorithmic}
\label{Protocol: One-shot}
\end{algorithm}

We now compare the two approaches. The one-shot method requires a single interferometer configuration, whereas the Fourier method requires multiple configurations. As discussed in Appendix~\ref{sec: Conditioning number}, random interferometers with at least as many modes as particles ($m\geq (1+\delta) n$ with $\delta>0$) yield well-conditioned matrices with high probability and, consequently, small inversion errors. Moreover, since it relies solely on classical post-processing of the collected samples, the one-shot method can serve as an efficient certification tool for Boson Sampling experiments. On the other hand, as shown later, its reconstruction error is larger than that of the Fourier method, so that more samples are needed to reach the same precision, which may be experimentally demanding.

\section{Numerical results}
\FloatBarrier

While we have proved appealing theoretical upper bounds, as bosonic linear optical experiments continue to develop, it is of interest to see how our reconstruction procedure works in practice. To this end we have applied our reconstruction algorithm with simulated samples using both the one-shot method, where all relevant Bargmann invariants are measured from one interferometer, and Fourier method where each relevant Bargmann is decomposed into factors given by the cycle decomposition of a permutation and each factor is measured using a dedicated Fourier interferometer. As shown in section~\ref{sec:estimators-error} the dependence of the Gram matrix reconstruction on the specific Bargmann estimation procedure is entirely due to the dependence of the quantity $A$ in the statement of Proposition~\ref{th:main-MSE-simple} on the $\sigma^2$ upper bounding the variances in equation~\eqref{eq:bargmann-variance}. For the two methods of estimating Bargmanns we have considered, the effective $\sigma$ is controlled by quite different aspects. The Fourier method uses a dedicated experiment to estimate each Bargmann invariant needed by the protocol, which means that the $N$ in the scaling $\sigma^2/N$ of equation~\eqref{eq:bargmann-variance} is reduced by the number of individual Bargmann invariant being estimated. The one-shot method, on the other hand does not suffer from this limitation, however there appears to be no ``free lunch'', and based on the numerical experiments we have performed the Bargmann invariants estimated with the one-shot method appear to have substantially larger variance for the same $N$ as those from the Fourier method, resulting from the inversion of the linear system. Overall these effects seem to lead to the Fourier method having better performance on the cases we have investigated numerically, despite the lower number of shots available for each Bargmann being estimated. In all of the figures in this section, where we plot upper bounds depending on the parameter $\sigma$ it is estimated empirically from data for the one-shot method and set to the upper bound $\sigma \leq 1$ valid for the Fourier method. In each case the one-shot method is implemented with a single Haar random unitary acting on twice as many modes as there are photons while the Fourier method spends half of the allocated experimental shots estimating the $n(n-1)/2$ second-order Bargmann invariants (i.e. the overlap of each state), and splits the remaining shot-budget evenly across the required higher order Bargmann invariants. In the case studied in Figure~\ref{fig:star-graph}, for example, the Fourier method first estimates $21$ second-order overlaps and $15$ third order Bargmann invariants each corresponding a path through photon $1$ connecting photons $i$, $j$ such that $1 < i < j \leq 7$. The error bars in each plot are determined from the sample variance across $10$ independent repetitions of each data-point. Throughout this section matrix $l^p$ distances refer to the Shatten norms.

\begin{figure*}[tbhp!]
    \centering
    \begin{subfigure}{0.49\textwidth}   \includegraphics[width=\linewidth]{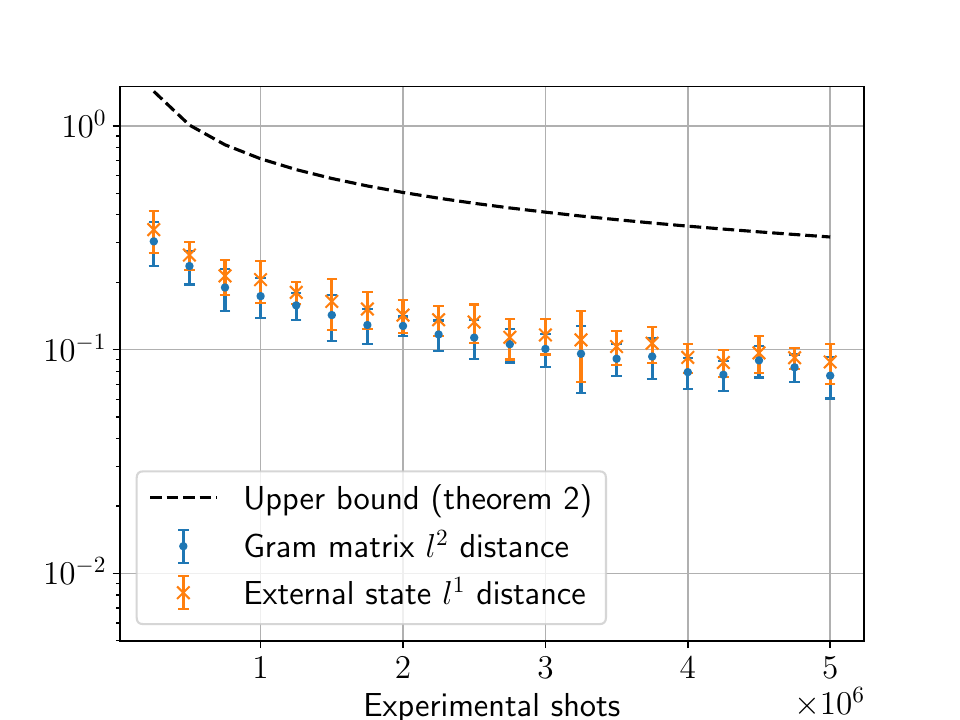}
    \caption{One shot method}
    \end{subfigure}
    \hfill
    \begin{subfigure}{0.49\textwidth}   \includegraphics[width=\linewidth]{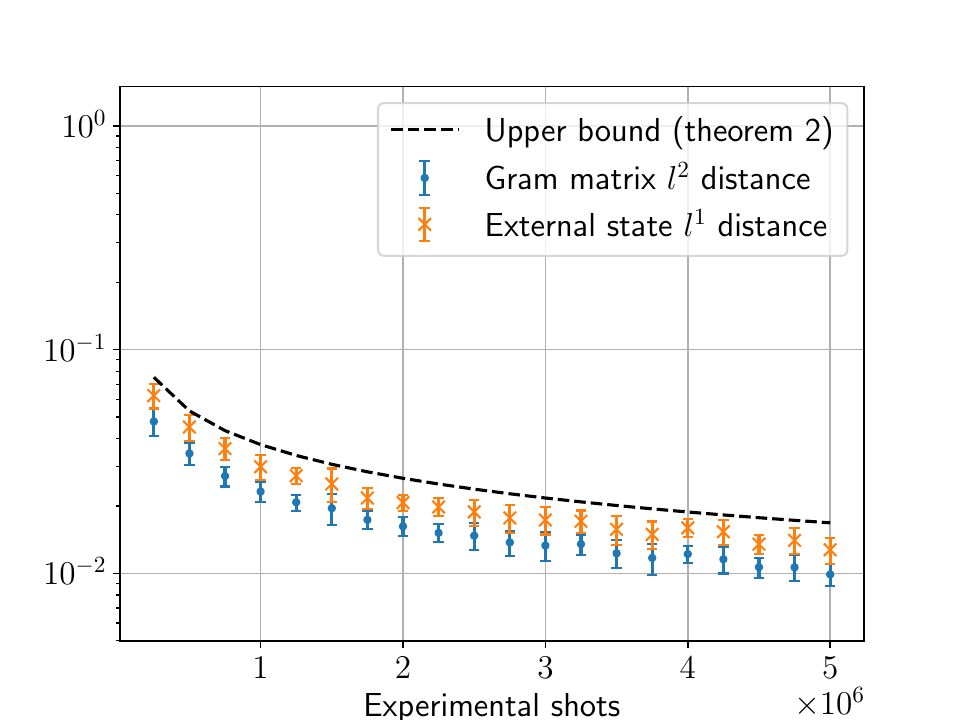}
    \caption{Fourier method}
    \end{subfigure}
    \caption{Comparison of our upper bound, the Frobenius distance between the recovered and true Gram matrices and the trace-norm distance between the recovered and true external states for a system of $7$ photons arranged in a ``star'' pattern. Photon $1$ is in pure state $\ket{1}$ while photons $j>1$ have internal states $x\ket{1} + \sqrt{1-x^{2}} \ket{j}$, where we fixed $x=0.9$. On the left (a) the reconstruction via the one-shot method is shown while the right (b) shows the Fourier method. The points and error bars plotted represent the mean and standard deviation in the Frobenius norm taken across $10$ independent reconstructions for each number of shots plotted.}
    \label{fig:star-graph}
\end{figure*}

Figures~\ref{fig:star-graph},~\ref{fig:qubit-sic} and~\ref{fig:qubit-mixed-main} include direct comparisons of the two methods. We see that, for the examples we have investigated, the performance of the one-shot method at the same number of experimental shots is noticeably worse than the one of the Fourier method. In an experimental context this should be contrasted with the lower experimental requirements for the one-shot method, rather than requiring specific, dedicated interferometers for each Bargmann measured, it can, in principle be performed using experimental data from \emph{any} single interferometer. Moreover we have theoretical evidence (see Theorem \ref{th:main-condition} in Appendix \ref{sec: Conditioning number}) that the one-shot method performs well with even a uniformly chosen random interferometer, complementing the numerical evidence of this section. 

In the cases we have compared the one-shot's performance to our theoretical upper bound, e.g. in Figures~\ref{fig:star-graph},~\ref{fig:qubit-sic} and~\ref{fig:qubit-mixed-main} we observe a significant ``gap'' between the theoretical bounds and the true performance. The theoretical guarantees appear to be rather pessimistic. One possible explanation for this is that the theoretical bounds take a single variance parameter $\sigma$, characterizing the Bargmann estimators. However the one-shot methods behaves slightly differently for the second and third order Bargmanns -- the second order Bargmanns are reconstructed first and the third order ones are reconstructed using a second linear system. It may be the case that these two sets of Bargmanns have differing variances, and using a single bound for both is overly coarse. Across both reconstruction methods we observe that Gram matrix tomography in practice may be performed with significantly fewer shots than our theoretical upper bounds would predict, with the ``in practice'' error at $250,000$ shots being comparable to the upper bound at $10$ times as many.
\begin{figure*}[tbhp!]
    \centering
    \begin{subfigure}{0.49\textwidth}   \includegraphics[width=\linewidth]{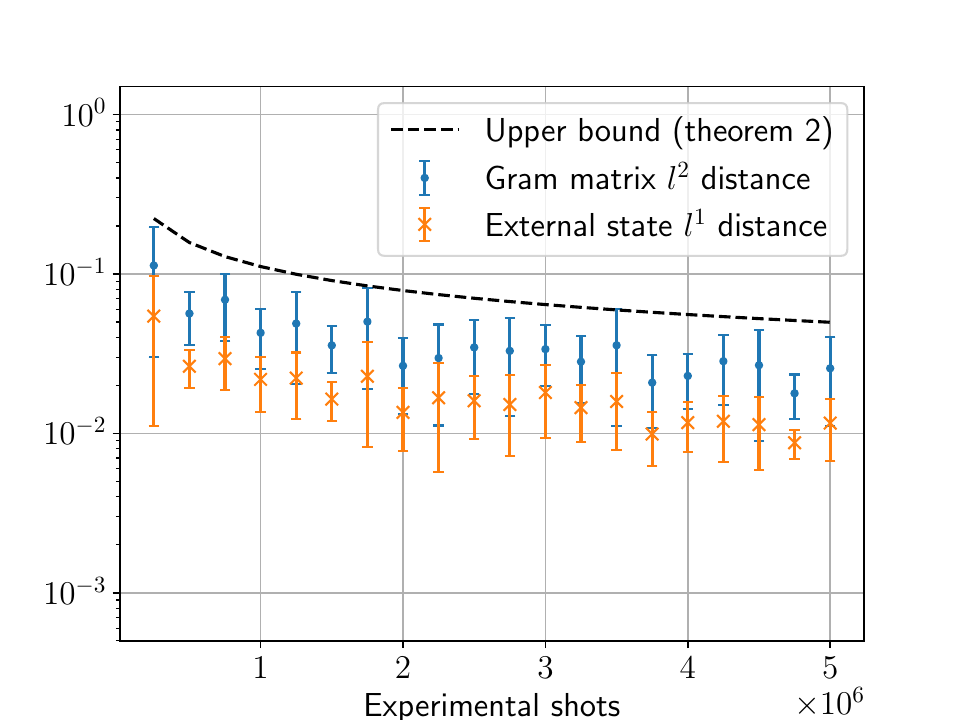}
    \caption{One-shot method}
    \end{subfigure}
    \hfill
    \begin{subfigure}{0.49\textwidth}   \includegraphics[width=\linewidth]{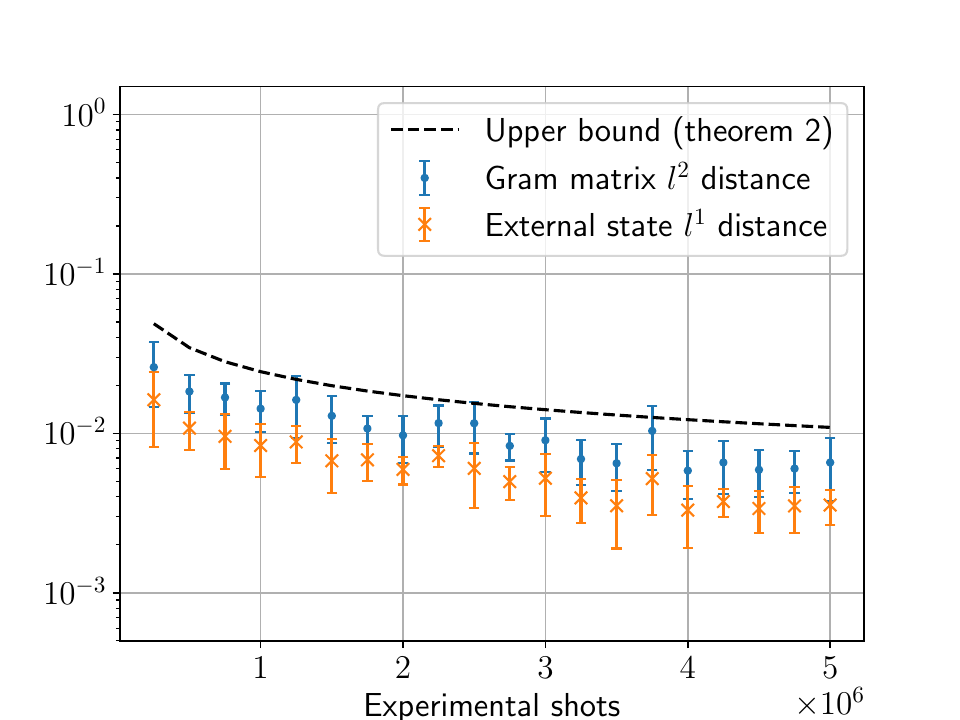}
    \caption{Fourier method}
    \end{subfigure}
    \caption{Comparison of our upper bound, the Frobenius distance between the recovered and true Gram matrices and the trace-norm distance between the recovered and true external states. Here there are $4$ photons whose single qubit states are the $4$ states whose projectors are proportional to the $4$ effects of a qubit SIC-POVM. For panel (a) we determine the Bargmann invariants with the one-shot method, while panel (b) make use of the Fourier method.}
    \label{fig:qubit-sic}
\end{figure*}

While we focus on the Gram matrix, as the object that determines the complete relational information, the external state $\Omega$ more directly specifies the outcome probabilities of measurements in linear optical experiments. In general the external state is defined as a factorial-sized matrix, so reconstruction rapidly becomes inefficient as the number of particles grows. Nonetheless the examples in this section, with up to $7$ photons per experiment, are small enough that we can reconstruct the external state, and therefore plot the quantities appearing in both sides of inequality~\eqref{eq:prop-gg2}. The trace-norm difference of external states is of both significant practical and theoretical relevance, as it controls the accuracy of the predictions one would make for the outcomes of measurements one can make on the external state. In particular, given a good approximation of the external state in trace-norm, one can obtain a good estimate of the probability distribution (in total variation distance) obtained by applying an \emph{arbitrary} interferometer and performing photon-number measurements on each mode.

The constellation of states in Figure~\ref{fig:qubit-sic} is notable due to the significant coherence between the states involved, forming a regular tetrahedron whose circumsphere is the Bloch sphere. Explicitly in the computational basis one can express the $4$ state vectors as
\begin{equation}
    \begin{aligned}
        \ket{\phi_1} &= \ket{0}, &
    \ket{\phi_2} &= \frac{\ket{0}+\sqrt{2}\ket{1}}{\sqrt{3}},\\ 
    \ket{\phi_3} &= \frac{\ket{0}+\sqrt{2}e^{i\frac{2\pi}{3}}\ket{1}}{\sqrt{3}},&
    \ket{\phi_4} &= \frac{\ket{0}+\sqrt{2}e^{i\frac{4\pi}{3}}\ket{1}}{\sqrt{3}}.
    \end{aligned}
\end{equation}
This constellation has significant applications in quantum information theory more broadly, as the states are (up to an overall normalization) equal to the effects of a qubit SIC POVM.  

\begin{figure*}[thbp!]
    \centering
    \begin{subfigure}{0.49\textwidth}   \includegraphics[width=\linewidth]{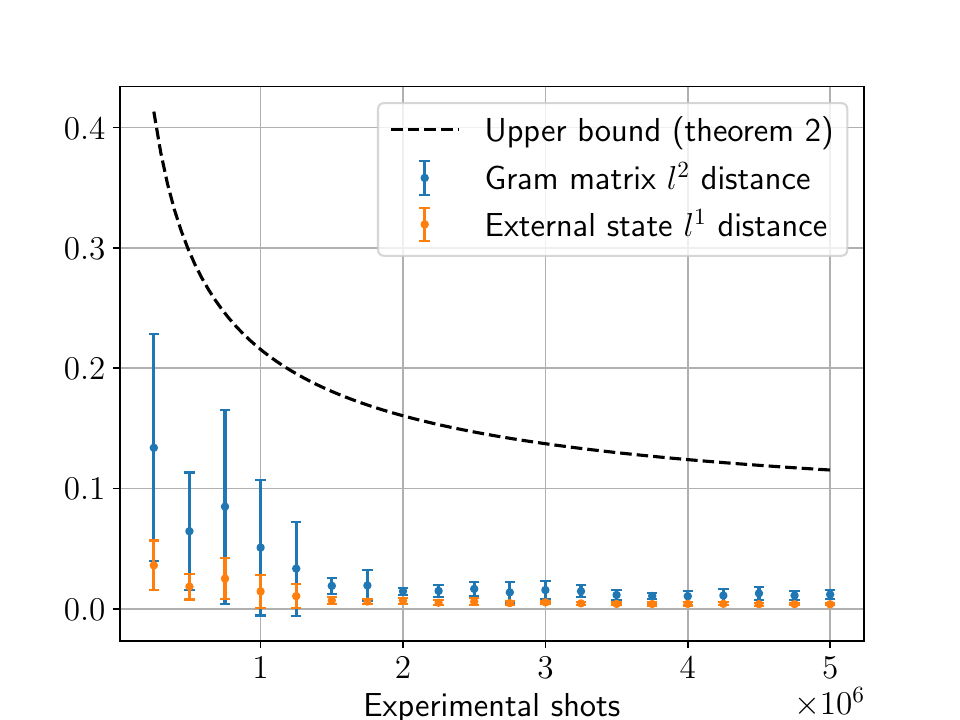}
    \end{subfigure}
    \hfill
    \begin{subfigure}{0.49\textwidth}   \includegraphics[width=\linewidth]{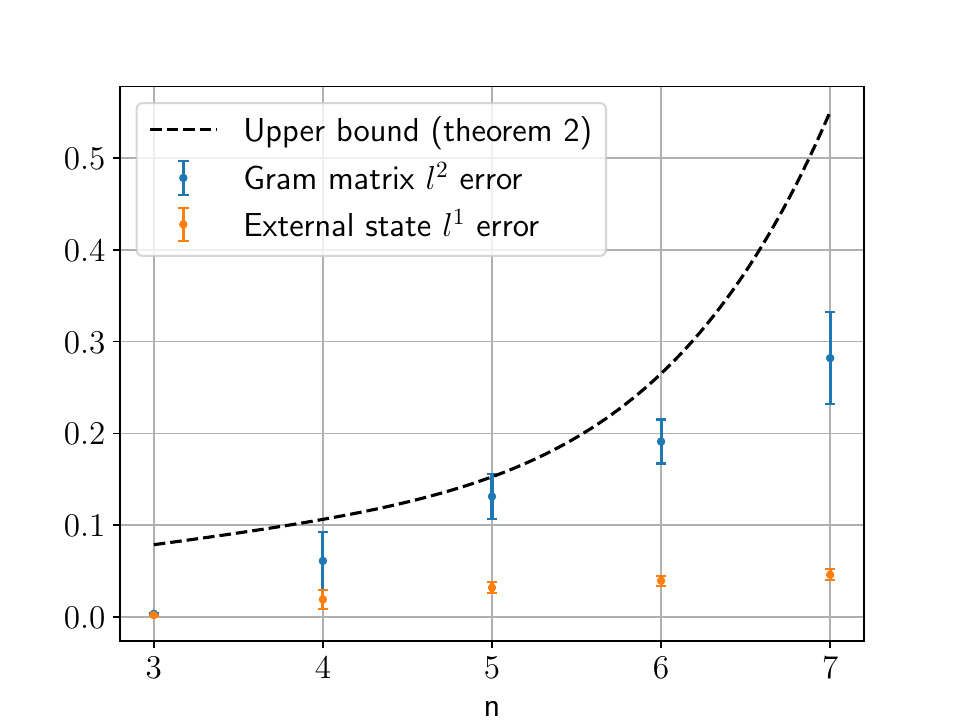}
    \end{subfigure}
    \caption{We apply our method, with Bargmann invariants measured using the Fourier method, to a system of ``circle dance states'', where we have $j=1\hdots n-1$ $\braket{\phi_j}{\phi_{j+1}} = \frac{1}{2}$ for $j=1\hdots n-1$ with non-trivial complex phase $\braket{\phi_1}{\phi_n} = \frac{i}{2}$. On the left we show the impact of increasing the number of experimental shots on a fixed-size system $n=5$, while on the right we fix the number of experimental shots at 5,000,000 and increase $n$. Our method is not expected to be efficient in the number of shots as $n$ grows, since we have to estimate an exponentially decaying Bargmann corresponding to the full cycle acting on all $n$ states. Here the threshold value $\eta$ is chosen to be 0.01.
 }
    \label{fig:circle-dance}
\end{figure*}
In Figure~\ref{fig:circle-dance} we apply our method to a setup in which the full relational information is impossible to recover by measuring Bargmanns of any fixed order less than the number of particles. In particular measuring visibilities and third order Bargmanns is not sufficient. The circle-dance states~\cite{shchesnovichCollectivePhasesIdentical2018} satisfy $\braket{\phi_j}{\phi_k} = 0$ unless $\abs{j-k} \leq 1$, and in our instantiation we also set $\braket{\phi_j}{\phi_{j+1}} = \frac{1}{2}$ for $j=1\hdots n-2$ and $\braket{\phi_{n-1}}{\phi_{n}} = \braket{\phi_{1}}{\phi_{n}} = \frac{i}{2}$. While our method is not expected to be efficient for states of this form as the number of particles grows, since one has to estimate the $B_\pi$ for the $\pi$ cyclically shifting $i \mapsto i+1 \text{ mod } n$ which decays exponentially in $n$, we observe performance significantly better than our theoretical upper bound implies, in particular the quantity which quantifies our ability to make predictions about the outcomes of photonic experiments, the $l_1$ error of the external states, remains very low even as the upper bound and the error in the Gram matrices grows. In Figure~\ref{fig:qubit-mixed-main} we explore the impact of having impure internal states on our protocol. The upper bound from theorem~\ref{th:main-MSE-simple} consists of two parts, one part which could be termed ``bias'' arising from the impurity of the states, and one part which could be called ``variance'' which decays with the number of experimental shots. As expected, in Figure~\ref{fig:qubit-mixed-main} as the number of experimental shots increases we observe our bound and the true observed errors in our reconstructions becoming asymptotic to the bias. Conversely as the single photon error rate becomes small at a fixed number of shots, we observe the quantities plotted becoming asymptotic to the ``variance'' term.

\begin{figure*}[tbhp!]
    \centering
    \begin{subfigure}{0.49\textwidth}   \includegraphics[width=\linewidth]{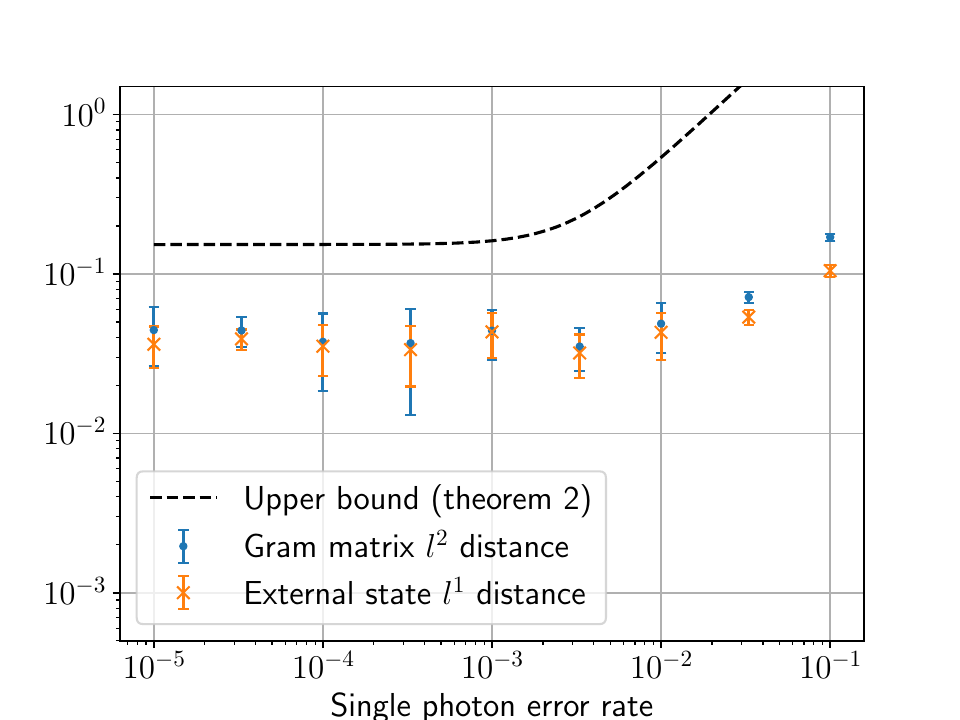}
    \caption{One shot method}\label{fig:qubit-mixed-impurity-oneshot}
    \end{subfigure}
    \hfill
    \begin{subfigure}{0.49\textwidth}   \includegraphics[width=\linewidth]{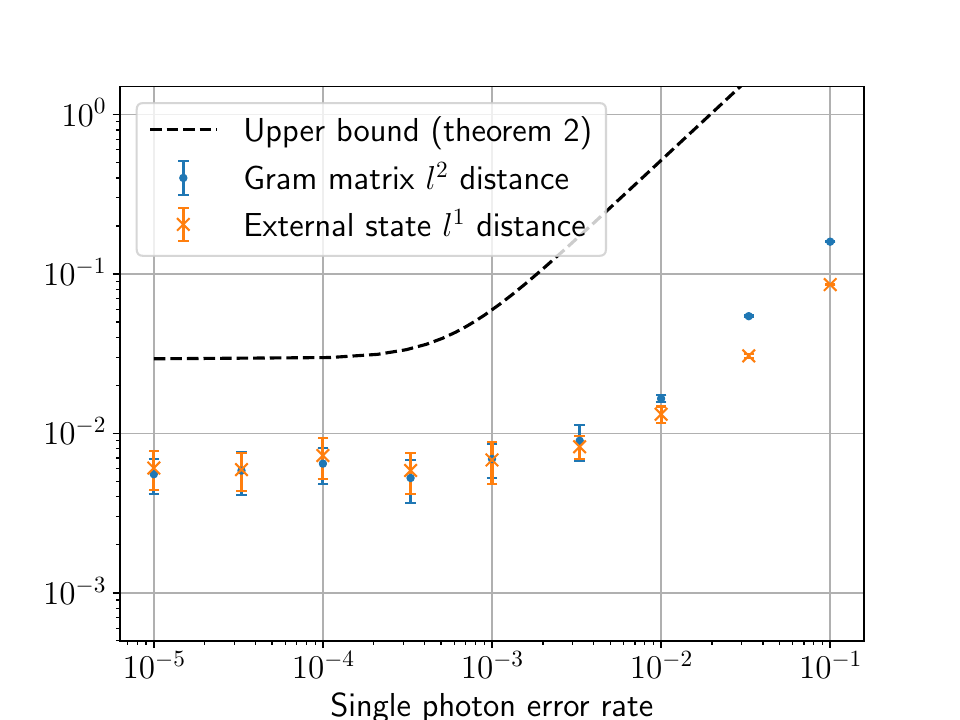}
    \caption{Fourier method}\label{fig:qubit-mixed-impurity-fourier}
    \end{subfigure}\\
    \begin{subfigure}{0.49\textwidth}   \includegraphics[width=\linewidth]{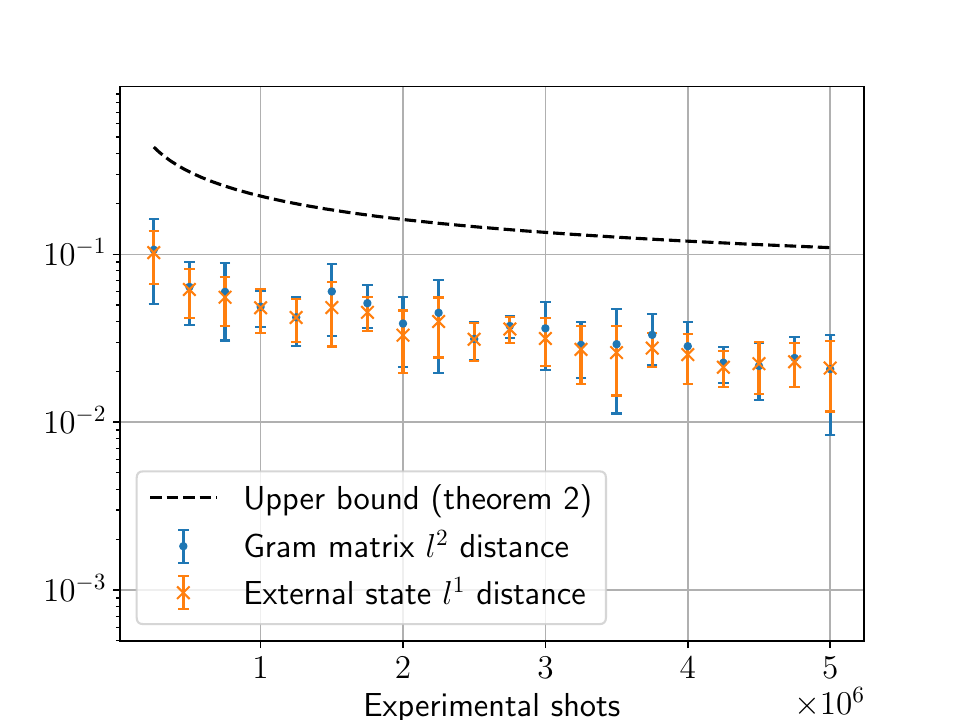}
    \caption{One shot method}\label{fig:qubit-mixed-shots-oneshot}
    \end{subfigure}
    \hfill
    \begin{subfigure}{0.49\textwidth}   \includegraphics[width=\linewidth]{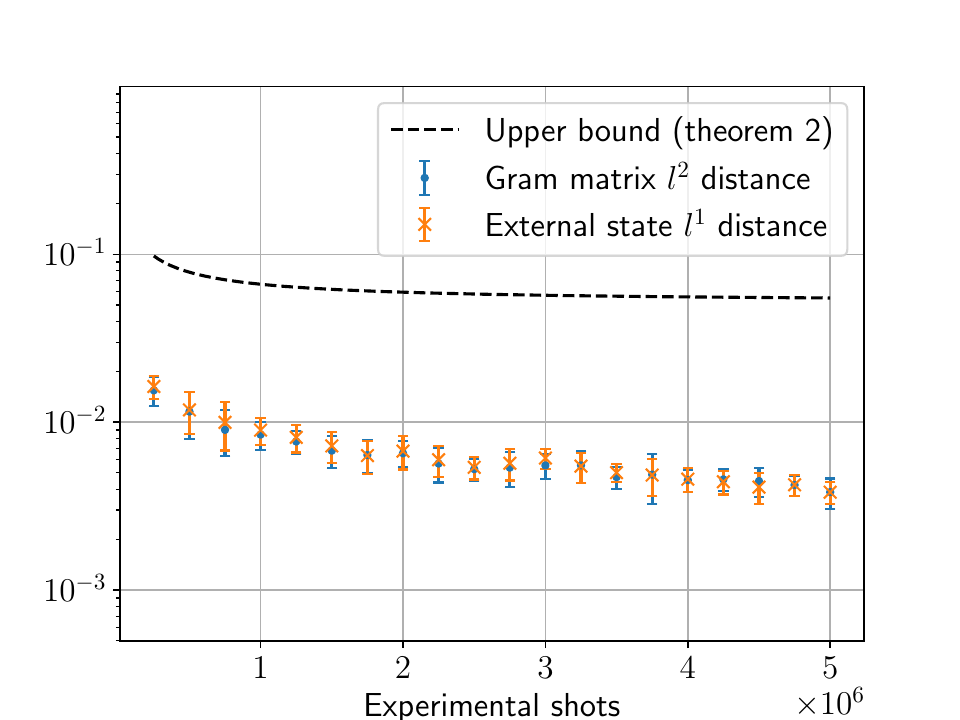}
    \caption{Fourier method}\label{fig:qubit-mixed-shots-fourier}
    \end{subfigure}
    \caption{Comparison of our upper bound, the Frobenius distance between the recovered and true Gram matrices and the trace-norm distance between the recovered and true external states for a system of $5$ qubit states. Each state is formed by first choosing a pure state uniformly at random within the circular cap on the Bloch sphere with Pauli $Z$ expectation-value greater than $0.9$. The single qubit error probability (the x-axis for subfigures~\ref{fig:qubit-mixed-impurity-oneshot} and~\ref{fig:qubit-mixed-impurity-fourier}) corresponds to depolarizing noise: each pure state is mixed with the maximally mixed state with to form $\rho_i = (1-p)\ketbra{\phi_i}{\phi_i} + p I/n$. For subfigures~\ref{fig:qubit-mixed-impurity-oneshot} and~\ref{fig:qubit-mixed-impurity-fourier} the total number of shots per simulated experiment is $2\times 10^6$, while for subfigures~\ref{fig:qubit-mixed-shots-oneshot} and~\ref{fig:qubit-mixed-shots-fourier} the single photon error rate is $3\times 10^{-3}$.}
    \label{fig:qubit-mixed-main}
\end{figure*}

\section{Conclusion}
 
In this work, we have introduced a framework for characterizing the partial distinguishability of photons with pure (or nearly pure) product internal states via Gram-matrix tomography. Rather than directly reconstructing the internal degrees of freedom, the protocol targets the Gram matrix describing their relational information, which fully determines the visible state $\Omegaext(\psi)$ (Proposition \ref{prop:tvd-continuity}) and hence the statistics of any experiment insensitive to the internal degrees of freedom (Corollary \ref{cor: total variation distance}). The resulting sample complexity is independent of the internal dimension and polynomial in $n$ whenever the required invariants are bounded from below, while we showed that $\Omega(n^{2}/\delta^{2})$ samples are necessary for any strategy.
 
We proposed two experimental routes to the required invariants: the Fourier method, which targets the estimation of each invariant with a dedicated interferometer, and the one-shot method, which extracts all second- and third-order invariants from the particle-number correlators of a single generic interferometer. Our analysis and simulations indicate that the one-shot method is well conditioned for random interferometers, but that the Fourier method achieves smaller errors for the same number of samples, reflecting a trade-off between experimental simplicity and reconstruction efficiency. In the cases considered, the observed errors lie below our analytical upper bounds.
 
Several questions remain open, including accounting for losses and calibration errors, and extending the framework beyond nearly pure product states, for which the minimal set of Bargmann invariants required for reconstruction has not yet been characterized. While we have introduced and explored two protocols for measuring the required Bargmann invariants for Gram matrix tomography in practical experiments, we do not claim that these protocols are optimal and it seems likely that future work could improve on these protocols.

We expect Gram-matrix tomography to become a practical tool for benchmarking photonic sources and for protocols that exploit photonic indistinguishability as a resource.

\section*{Acknowledgments }
 MR is a FRIA grantee of the Fonds de la Recherche Scientifique – FNRS. MO, MK, OR-S, BS and LN acknowledge the support from the European Union’s Horizon Europe research and innovation program under EPIQUE Project GA No. 101135288. The C4QEC project is carried out within the IRAP of the Foundation for Polish Science co-financed by the European Union. E. F. G.  acknowledges funding from the National Council for Scientific and Technological Development – CNPq (Brazil) under grant 308292/2025-1.

\section*{Code availability}

The code is available at the GitHub repository: \href{https://github.com/benoitseron/gram_tomography}{\texttt{https://github.com/benoitseron/gram\_tomography}} and relies on the following package \cite{seron2024bosonsampling}.

\section*{AI use disclosure}

ChatGPT 5.6 Sol was used for auxilliary computations in Appendix C and D. It also provided the proof of \Cref{prop:tvd-continuity} and \Cref{th:sample-lower-bound}. 
Claude Code was used to help write the codebase and plots.

\bibliography{Reference}
\clearpage
\onecolumngrid

\appendix
\section*{Appendix}

This appendix contains the technical proofs, statistical error bounds, and mathematical derivations supporting the main results. Appendix~\ref{sec: proof of lemma} provides proofs for the trace-distance upper bound and the sample complexity lower bound, Appendix~\ref{sec:estimators} analyzes the estimator error and the optimal spanning-tree selection for the reconstruction protocol, Appendix~\ref{sec: Conditioning number} proves the conditioning guarantees for random interferometers using Weingarten calculus and unitary concentration, Appendix~\ref{sec: Correlator derivation} details the derivation of particle-number correlators, Appendix~\ref{Sec: Certification} addresses the certification of multiphoton indistinguishability and Appendix ~\ref{sec: size reduction} shows how to reduce the linear system dimension for higher order invariants.

\section{Proofs of technical statements}\label{sec: proof of lemma}

\subsection{Trace distance upper bound (proof of \Cref{prop:tvd-continuity})}\label{sec:proof-of-tvd}

\begin{proof}
    In the permutation basis $\ket{\sigma}=\ket{\sigma(1),\ldots,\sigma(n)}=P^{\mathrm{ext}}_{\sigma^{-1}}\ket{\mathbf{i}_0}$, define the auxiliary $n!\times n!$ matrix $J(G)$ by $J(G)_{\sigma,\tau}:=\prod_{i=1}^{n}G_{\sigma(i),\tau(i)}$. Equation~\eqref{eq:omegaext-gram} then reads
    \begin{equation}
        \Omegaext(\psi) = \frac{1}{n!}\sum_{\sigma,\tau \in S_n}\ket{\tau}\bra{\sigma}J(G)_{\sigma,\tau}.
    \end{equation}
    Thus the matrix of $\Omegaext(\psi)$ in this basis is $J(G)^{T}/n!$. Transposition preserves the trace norm, so it suffices to bound $\norm{J(G)-J(G')}_1/n!$.
    Let $H = G-G'$ and for $k=1,\dots,n$ let $H_k$ be the $n! \times n!$ matrix defined by $(H_k)_{\sigma,\tau} := H_{\sigma(k),\tau(k)}$. Using a standard telescoping argument we can write:
    \begin{equation}
        J(G)_{\sigma,\tau}-J(G')_{\sigma,\tau} = \prod_{i=1}^{n}G_{\sigma(i),\tau(i)} - \prod_{i=1}^{n}G'_{\sigma(i),\tau(i)} = 
        \sum_{k=1}^{n}\left(\prod_{i=1}^{k-1}G'_{\sigma(i),\tau(i)} \right)(H_k)_{\sigma,\tau}\left(\prod_{i=k+1}^{n}G_{\sigma(i),\tau(i)}\right)
    \end{equation}
    Let us define:
    \begin{equation}
        (C_k)_{\sigma,\tau} := \prod_{i=1}^{k-1}G'_{\sigma(i),\tau(i)} \prod_{i=k+1}^{n}G_{\sigma(i),\tau(i)}
    \end{equation}
    so that:
    \begin{equation}
        J(G)_{\sigma,\tau}-J(G')_{\sigma,\tau} = \sum_{k=1}^{n}C_k \circ H_k
    \end{equation}
    where $\circ$ denotes the entrywise product. We then clearly have:
    \begin{equation}
        \norm{\Omegaext(\psi) - \Omegaext(\psi')}_1 = \frac{1}{n!}\norm{J(G)-J(G')}_{1} \leq
        \frac{1}{n!}\sum_{k=1}^{n}\norm{C_k \circ H_k}_{1}
    \end{equation}
    We now claim that each $C_k$ is a correlation matrix, i.e. a Gram matrix of unit vectors. To see this, choose any vectors $v_i, v'_i$ such that $G_{ij} = \langle v_i, v_j \rangle$ and $G'_{ij} = \langle v'_i, v'_j \rangle$. Then for each $\sigma \in S_n$ define:
    \begin{equation}
        w^{(k)}_{\sigma} := \bigotimes_{i=1}^{k-1}v'_{\sigma(i)}  \otimes\bigotimes_{i=k+1}^{n}v_{\sigma(i)}
    \end{equation}
    It is easily checked that each $w^{(k)}_{\sigma}$ has unit norm and the Gram matrix of this set of vectors is exactly $C_k$. We now use the standard fact (see e.g. \cite{watrous}) that entry-wise multiplication by a correlation matrix is contractive on the set of Hermitian matrices in the trace norm, i.e. $\norm{C \circ A}_{1} \leq \norm{A}_{1}$. We thus obtain:
    \begin{equation}
        \norm{\Omegaext(\psi) - \Omegaext(\psi')}_1 \leq
        \frac{1}{n!}\sum_{k=1}^{n}\norm{H_k}_{1}
    \end{equation}
    We will now show that each $\norm{H_k}_{1}$ is equal to $(n-1)! \cdot \norm{H}_{1}$, which will finish the proof since $H=G-G'$. To see this, fix $k$ and let us order the permutation basis $\{\ket{\pi}\}_{\pi \in S_n}$ according to the values of $\pi(k)$. The matrix $H_k$ in this ordering consists of consecutive blocks of size $(n-1)! \times (n-1)!$, where each block has entries equal to $H_{\sigma(k),\tau(k)}$. We can thus write $H_k = H \otimes J$, where $J$ is the $(n-1)! \times (n-1)!$ matrix of all ones. It follows that $\norm{H_k}_{1} = \norm{H}_{1} \cdot \norm{J}_{1} = (n-1)! \norm{H}_{1}$ as desired.
\end{proof}

\subsection{Sample complexity lower bound (proof of \Cref{th:sample-lower-bound})}\label{sec:proof-of-sample-lower-bound}

We start with an informal sketch. The logic of the lower bound proof follows a similar template as the lower bound on the sample complexity of tomography of quantum states \cite{haahSampleOptimalTomography2017}. Given a $\delta$, we construct a large family of states $\{\psi_G\}_{G \in \mathcal{G}}$, indexed by Gram matrices $G$ from a family $\mathcal{G}$ of size $2^{\Omega(n^2)}$, such that:
\begin{enumerate}
    \item External states $\Omegaext(\psi_G), \Omegaext(\psi_{G'})$ for $G, G' \in \mathcal{G}, G \neq G'$ have large mutual distance:
    \begin{equation}
        \norm{\Omegaext(\psi_G) - \Omegaext(\psi_{G'})}_1 > \delta/16
    \end{equation}
    \item Each external state $\Omegaext(\psi_G)$ has nearly minimal purity:
    \begin{equation}
        \tr(\Omegaext(\psi_G)^2) \leq \frac{1 +  2 \delta^2}{n!}
    \end{equation}
\end{enumerate}
where $n!$ is the dimension of the total space in which $\Omegaext(\psi_G)$ live. Property 1 allows us to reduce the task of external state estimation to state identification: given $N$ copies of $\Omegaext(\psi_G)$ for unknown $G \in \mathcal{G}$, output the correct label $G$. Performing such identification with high probability requires at least $\approx \log\vert \mathcal{G} \vert = \Omega(n^2)$ bits of information about the label. Meanwhile, property 2 implies that each state $\Omegaext(\psi_G)$ has nearly maximal entropy. This, in turn, implies that the label information is only weakly correlated with the state -- performing any measurement on a single copy can yield at most $O(\delta^2)$ bits of information about the label. Taken together this implies that at least $N \approx \frac{n^2}{\delta^2}$ copies are needed.

For a graph $\Gamma$ on $n$ vertices let $A_{\Gamma}$ denote its adjacency matrix. For a fixed constant $\alpha$ we define the matrix:
\begin{equation}\label{eq:ggamma-def}
    G_{\Gamma} := I + \sqrt{\frac{\alpha}{n}}A_{\Gamma}
\end{equation}

\begin{lemma}\label{lm:half-of-the-graphs}
    For any $\alpha < 1/4$, for at least half of all possible graphs the matrix $G_{\Gamma}$ is positive semidefinite and thus a Gram matrix of $n$ vectors.
\end{lemma}

\begin{proof}
Write $A_{\Gamma} = \frac{J-I+X}{2}$, where $J$ is the all ones matrix and $X$ is a symmetric matrix with zero diagonal and $X_{ij}=X_{ji} = \pm 1$ with probability $1/2$. Since $J \succeq 0$, the condition that $G_{\Gamma} \succeq 0$ is easily seen to be equivalent to:
\begin{equation}\label{eq:operator-norm-graph}
     X \succeq -\left(2\sqrt{\frac{n}{\alpha}}-1\right)I
\end{equation}
It follows from classic random matrix theory \cite{FurediKomlos1981} that the operator norm of $X$ satisfies $\norm{X}_{\infty} \leq 4\sqrt{n}$ with probability tending to $1$ as $n\to\infty$. A more precise nonasymptotic statement \cite{Vu2007} gives \eqref{eq:operator-norm-graph} for $\alpha < 1/4$ 
\end{proof}

\begin{lemma}\label{lm:large-family}
    Fix $\alpha < 1/4$. There exists a family of graphs $\mathcal{F}$ of size $\geq 2^{\frac{1}{10} \binom{n}{2}}$ such that: 1) for each $\Gamma \in \mathcal{F}$ the matrix $G_{\Gamma}$ is a Gram matrix; 2) for any $\Gamma, \Gamma' \in \mathcal{F}, \Gamma \neq \Gamma'$ the graphs differ in at least $\frac{1}{4}\binom{n}{2}$ edges.
\end{lemma}

\begin{proof}
    By \Cref{lm:half-of-the-graphs}, for at least half of all the $2^{\binom{n}{2}}$ graphs the matrix $G_{\Gamma}$ is a Gram matrix. Denote this family of graphs by $\mathcal{F}_{+}$. We construct the desired family $\mathcal{F}$ greedily -- start with any graph $\Gamma_1 \in \mathcal{F}_{+}$, remove from $\mathcal{F}_{+}$ all graphs which differ from $\Gamma_1$ in less than $\frac{1}{4}\binom{n}{2}$ edges, pick any graph from the remaining family and iterate until no more graphs can be picked. The classical Gilbert-Varshamov bound \cite{HuffmanPless2003} guarantees that the size of obtained family $\mathcal{F}$ is at least $2^{\binom{n}{2}(1 - h(1/4))-1} \geq 2^{\frac{1}{10}\binom{n}{2}}$, where $h$ is the binary entropy function.
\end{proof}

We now define $\mathcal{G} := \{ G_{\Gamma} \ \vert \ \Gamma \in \mathcal{F} \}$, where $\mathcal{F}$ is the graph family constructed in \Cref{lm:large-family}. For any $G \in \mathcal{G}$ we consider the pure product internal state $\psi_G$ with Gram matrix $G$.

\begin{lemma}\label{lm:purity-and-distance}
    For any $G \in \mathcal{G}$ the state $\Omegaext(\psi_G)$ satisfies:
    \begin{equation}
        \tr(\Omegaext(\psi_G)^2) < \frac{1 + 2\alpha^2}{n!}\label{eq:purity-bound}
    \end{equation}
    Moreover, for any $G, G' \in \mathcal{G}, G \neq G'$ we have:
    \begin{equation}\label{eq:tr-distance-gg'}
        \norm{\Omegaext(\psi_G) - \Omegaext(\psi_{G'})}_1 > \frac{\alpha}{16}
    \end{equation}
\end{lemma}

\begin{proof}
We first prove the purity bound \eqref{eq:purity-bound}. By \eqref{eq:omegaext-gram} we can write the purity as:
\begin{equation}
     \tr(\Omegaext(\psi_G)^2) = \frac{1}{n!}\sum_{\sigma \in S_n}\abs{\prod_{i=1}^{n} G_{i,\sigma(i)} }^2
\end{equation}
If a permutation $\sigma \in S_n$ displaces exactly $k$ vertices, by \eqref{eq:ggamma-def} its corresponding term in the sum is either zero (if for any $i$ the edge $(i,\sigma(i))$ is not present in the graph) or $(\frac{\alpha}{n})^k$. The number of such permutations is at most $k! \cdot \binom{n}{k} \leq n^k$. Thus:
\begin{equation}
     \tr(\Omegaext(\psi_G)^2) \leq \frac{1}{n!}\left(1 + \sum_{k=2}^{n}n^k \cdot \left(\frac{\alpha}{n}\right)^k\right) \leq \frac{1}{n!}\left(1 + \frac{\alpha^2}{1-\alpha}\right) \leq \frac{1 + 2\alpha^2}{n!}
\end{equation}

We now show \eqref{eq:tr-distance-gg'}. Take $G \neq G', G, G' \in \mathcal{G}$ corresponding to two different graphs $\Gamma, \Gamma'$. By \Cref{lm:large-family} the symmetric difference of the edge sets $\mathcal{E} :=E(\Gamma) \Delta E(\Gamma')$ has size at least $\frac{1}{4}\binom{n}{2} > \frac{1}{12}n^2$. For an edge $(i,j) \in \mathcal{E}$ define the sign:
\begin{equation}
    s_{ij} := (A_{\Gamma})_{ij}-(A_{\Gamma'})_{ij}
\end{equation}
and:
\begin{equation}
    T := \sum_{i<j,(ij) \in \mathcal{E}} s_{ij}L_{(ij)}
\end{equation}
where $L_{\pi}\ket{\sigma}:=\ket{\pi\circ\sigma}$ permutes the external mode labels. In particular, $L_{(ij)}$ exchanges mode labels $i$ and $j$ in every tensor factor. Let $D := \Omegaext(\psi_G) - \Omegaext(\psi_{G'})$, so that $ \norm{\Omegaext(\psi_G) - \Omegaext(\psi_{G'})}_1  = \tr(\abs{D})$. We claim that:
\begin{equation}\label{eq:noncomm-berger}
    \tr(\abs{D}) \geq \frac{\abs{\tr(DT)}^2}{\sqrt{\tr(D^2) \tr(T^4)}}
\end{equation}
Indeed, applying Cauchy-Schwarz twice gives:
\begin{equation}
    \abs{\tr(DT)}^2 \leq \tr(\abs{D})\tr(\abs{D}T^2) \leq  \tr(\abs{D})\sqrt{\tr(D^2) \tr(T^4)}
\end{equation}
The inequality \eqref{eq:noncomm-berger} can be seen as a noncommutative variant of Berger's inequality, which states that $\E\abs{X} \geq \frac{(\E X^2)^{3/2}}{(\E X^4)^{1/2}}$ for a scalar random variable $X$. 

We now bound the three quantities appearing in \eqref{eq:noncomm-berger}. The trace $\tr(DT)$ can be computed exactly using orthogonality of $\ket{\sigma}$:
\begin{equation}
    \tr(\Omegaext(\psi_G)T) =  \frac{1}{n!}\tr\left(\sum_{\substack{\sigma,\tau \in S_n \\ (ij) \in \mathcal{E}}}\ket{\tau}\bra{\sigma}s_{ij}L_{(ij)} \prod_{k=1}^{n}G_{\sigma(k),\tau(k)}\right) = 
    \frac{1}{n!}\tr\left(\sum_{\substack{\sigma \in S_n \\ (ij) \in \mathcal{E}}} s_{ij} \prod_{k=1}^{n}G_{\sigma(k),((ij)\sigma)(k)}\right) = \sum_{(ij)\in \mathcal{E}}s_{ij}G^2_{ij}
\end{equation}
Substracting the trace gives:
\begin{equation}\label{eq:trdt}
    \tr(DT) = \sum_{(ij)\in \mathcal{E}}s_{ij}(G^2_{ij} - (G')^2_{ij}) = \sum_{(ij)\in \mathcal{E}}s^2_{ij}\frac{\alpha}{n} =\frac{\alpha}{n}\abs{\mathcal{E}} 
\end{equation}
The trace $\tr(D^2)$ is computed similarly as the purity:
\begin{equation}\label{eq:trd2}
    n! \tr(D^2) = \sum_{\sigma \in S_N}\left( \prod_{k=1}^{n}G_{i,\sigma(i)} - \prod_{k=1}^{n}G'_{i,\sigma(i)}\right)^2 \leq \frac{\alpha^2}{n^2}\abs{\mathcal{E}} + \sum_{k=3}^{n}\alpha^k < \frac{\alpha^2}{n^2}\abs{\mathcal{E}} + \frac{\alpha^3}{1-\alpha} < \frac{5\alpha^2}{n^2}\abs{\mathcal{E}} 
\end{equation}
where in the last inequality we have used $\alpha < 1/4, \abs{\mathcal{E}} > \frac{n^2}{12}$. The first inequality follows by splitting the sum over $\sigma$ intro transpositions and the rest. The latter is bounded exactly as in the computation for the purity, while the former is treated as in \eqref{eq:trdt}.

Finally, we need to upper bound $\tr(T^4)$. To this end, we count the number of ordered quadruples of the $\abs{\mathcal{E}}$ transpositions whose product is the identity. When the first two coincide, there are $\abs{\mathcal{E}}^2$ possibilities. Otherwise, their product is either a three-cycle or two disjoint transpositions and its inverse has at most three ordered decompositions into two transpositions. Hence:
\begin{equation}\label{eq:trt4}
\tr(T^4) = n! (\abs{\mathcal{E}}^2 + 3\abs{\mathcal{E}}(\abs{\mathcal{E}}-1)) < n!\cdot4 \abs{\mathcal{E}}^2    
\end{equation}
Plugging \eqref{eq:trdt}, \eqref{eq:trd2} and \eqref{eq:trt4} into \eqref{eq:noncomm-berger} yields:
\begin{equation}
    \tr(\abs{D}) \geq \frac{(\frac{\alpha}{n}\abs{\mathcal{E}})^2}{(\frac{5\alpha^2}{n^2}\abs{\mathcal{E}} \cdot  4 \abs{\mathcal{E}}^2  )^{1/2}} = \frac{\alpha \abs{\mathcal{E}}^{1/2}}{2\sqrt{5}n} > \frac{\alpha}{16}
\end{equation}
 as desired.

\end{proof}

\begin{proof}[Proof of \Cref{th:sample-lower-bound}]
    Consider the family of states $S = \{\Omegaext(\psi_G))\}_{G \in \mathcal{G}}$ for $\alpha = 32\delta$. Consider any estimator $\widehat{\rho}$ that receives $N$ copies of the input state and satisfies for any input $G \in \mathcal{G}$:
    \begin{equation}
        \mathbb{P}\left( \norm{\widehat{\rho} - \Omegaext(\psi_G)}_1 < \delta \right) > \frac{2}{3}
    \end{equation}
    The estimator can be used to solve the state identification task for $S$: given an unknown state $\Omegaext(\psi_G))$ with uniform probability, return the correct label $G$. We first apply the estimator $\widehat{\rho}$ to the input state and then return the label $G$ of the state from $S$ closest to $\widehat{\rho}$. By \Cref{lm:purity-and-distance} any two different states $\Omegaext(\psi_G), \Omegaext(\psi_{G'})$ are at least $2\delta$ apart, so if the input is $\Omegaext(\psi_G)$, the procedure gives the correct label $G$ with probability at least $2/3$. Thus, it suffices to show that solving the state identification problem for $S$ with probability $\geq 2/3$ requires at least $N = \Omega(\frac{n^2}{\delta^2})$ samples.

    To bound the maximum probability for the state identification problem, we invoke Lemma 1 from \cite{OgawaNagaoka1999} for the family of states $\{ \Omegaext(\psi_G)^{\otimes N} \}_{G \in \mathcal{G}}$ and $\beta = 1/2$. In our notation the lemma states that the maximum probability of correct state identification is bounded by:
    \begin{equation}
        p \leq \frac{1}{\abs{\mathcal{G}}}\tr \left( \sum_{G \in \mathcal{G}} (\Omegaext(\psi_G)^{\otimes N})^2 \right)^{1/2}
    \end{equation}
    For any positive semidefinite matrix $A$ on a space of dimension $d$ we have $\tr\sqrt{A} \leq \sqrt{d} \sqrt{\tr A}$. By applying \Cref{lm:purity-and-distance} we get:
\begin{equation}
        p\leq\frac{1}{\abs{\mathcal{G}}}\tr \left( \sum_{G \in \mathcal{G}} (\Omegaext(\psi_G)^{\otimes N})^2 \right)^{1/2} \leq 
        \frac{1}{\abs{\mathcal{G}}} \sqrt{(n!)^N} \left( \sum_{G \in \mathcal{G}} \tr((\Omegaext(\psi_G)^{\otimes N})^2) \right)^{1/2} \leq
        \sqrt{\frac{(1 + 2\cdot 32^2 \delta^2)^N}{\abs{\mathcal{G}}}} 
    \end{equation}
Since we assumed $p \geq 2/3$, we obtain:
\begin{equation}
    4/9 \abs{\mathcal{G}} \leq (1 + 2048 \cdot \delta^2)^N
\end{equation}
so:
\begin{equation}
    N > \frac{\log(4/9 \abs{\mathcal{G}})}{\log(1 + 2048 \cdot \delta^2)} > \frac{\log(4/9 \abs{\mathcal{G}})}{2048\cdot \delta^2} 
\end{equation}
Recalling from \Cref{lm:large-family} that $\abs{\mathcal{G}} \geq 2^{\frac{1}{10}\binom{n}{2}}$ yields the desired sample lower bound.
\end{proof}

\section{Analysis of the tomography protocol}\label{sec:estimators}

\subsection{Error of estimators}\label{sec:estimators-error}






Suppose for now that we have fixed a deterministic spanning tree $T$. We first analyze the estimator $\widehat{G}_{ij}$ for $(i,j) \in T$.

\begin{proposition}\label{prop:inT}
    For $(i,j) \in T$ the mean square error of $\widehat{G}_{ij}$ given by \eqref{eq:gram1} is bounded by:
\begin{equation}
    \E\abs{\widehat{G}_{ij} - G_{ij}}^2 \leq \frac{1}{B_{ij}}\cdot \left(16\varepsilon^2 + \frac{k\sigma^2}{N}\right)
\end{equation}\end{proposition}

\begin{proof}
We have:
\begin{equation}
    \abs{\sqrt{\widehat{B_{ij}}}-\sqrt{B_{ij}}} = \frac{\abs{\widehat{B_{ij}} - B_{ij}}}{\sqrt{\widehat{B_{ij}}}+\sqrt{B_{ij}}} \leq \frac{\abs{\widehat{B_{ij}} - B_{ij}}}{\sqrt{B_{ij}}}
\end{equation}
and therefore by \eqref{eq:bargmann-MSE}:
\begin{equation}
    \E\abs{\widehat{G}_{ij} - G_{ij}}^2 \leq \frac{1}{B_{ij}}\cdot \E\abs{\widehat{B_{ij}} - B_{ij}}^2 \leq
    \frac{1}{B_{ij}}\cdot \left(16\varepsilon^2 + \frac{k\sigma^2}{N}\right)
\end{equation}
\end{proof}


We now analyze the more complicated estimator $\widehat{G}_{ij}$ for $(i,j) \notin T$ given by \eqref{eq:gram2}. 


\begin{proposition}\label{prop:notinT}
     For $(i,j) \notin T$ the mean square error of $\widehat{G}_{ij}$ given by \eqref{eq:gram2} is bounded by:
\begin{equation}
    \E (\widehat{G_{ij}} - G_{ij})^2 \leq
    \frac{2}{B_{ij}} \left( 16\varepsilon^2 + \frac{k\sigma^2}{N} \right) +  \frac{8B_{ij}}{\abs{B_{\pi(i,j)}}^2} \cdot \left(4m^2_{ij}\varepsilon^2 + \frac{k\sigma^2}{N}   \right) 
\end{equation}    
where $m_{ij}$ is one plus the length of the path connecting $i$ and $j$ in the tree $T$.
\end{proposition}

\begin{proof}
For notational simplicity let $X = e^{i\varphi_{ij}}, Y = \sqrt{B_{ij}}, \widehat{X} = \widehat{e^{i\varphi_{ij}}}, \widehat{Y} = \widehat{\sqrt{B_{ij}}}$, so that $G_{ij} = XY, \widehat{G_{ij} }= \widehat{X}\cdot \widehat{Y}$. Recall that $X = \frac{B_{\pi(i,j)}}{\abs{B_{\pi(i,j)}}}$ and $\widehat{X} =  \frac{\widehat{B}_{\pi(i,j)}}{\abs{\widehat{B}_{\pi(i,j)}}}$. We first write:
\begin{equation}
\widehat{G_{ij}} - G_{ij} =  \widehat{X}\cdot\widehat{Y}- XY  = \widehat{X}(\widehat{Y} - Y) + Y(\widehat{X} - X)
\end{equation}
so that:
\begin{equation}\label{eq:Gij-MSE}
    (\widehat{G_{ij}} - G_{ij})^2 \leq 2\abs{\widehat{X}}^2(\widehat{Y} - Y)^2 + 2Y^2\abs{\widehat{X} - X}^2 = 
    2(\widehat{Y} - Y)^2 + 2Y^2\abs{\widehat{X} - X}^2
\end{equation}
For any complex numbers $z, w$ with $w \neq 0$ the following inequality is easily verified:
\begin{equation}
    \abs{\frac{z}{\abs{z}} - \frac{w}{\abs{w}}} \leq \frac{2\abs{z-w}}{\abs{w}}
\end{equation}
This immediately yields:
\begin{equation}
    \abs{\widehat{X} - X} \leq \frac{2\abs{ \widehat{B}_{\pi(i,j)} - B_{\pi(i,j) }}}{\abs{B_{\pi(i,j)}}}
\end{equation}
Plugging back to \eqref{eq:Gij-MSE} we obtain:
\begin{equation}
    \E (\widehat{G_{ij}} - G_{ij})^2 \leq
    2\E(\widehat{B_{ij}} - B_{ij})^2 +  \frac{8B_{ij}}{\abs{B_{\pi(i,j)}}^2} \cdot \E\abs{ \widehat{B}_{\pi(i,j)} - B_{\pi(i,j) }}^2
\end{equation}
which by \eqref{eq:bargmann-MSE} gives:
\begin{equation}
    \E (\widehat{G_{ij}} - G_{ij})^2 \leq
    \frac{2}{B_{ij}} \left( 16\varepsilon^2 + \frac{k\sigma^2}{N} \right) +  \frac{8B_{ij}}{\abs{B_{\pi(i,j)}}^2} \cdot \left(4m^2_{ij}\varepsilon^2 + \frac{k\sigma^2}{N}   \right) 
\end{equation}
as desired.
\end{proof}

Consider now the protocol from \Cref{sec:protocol-main}. In order to use the above bounds we assume that the spanning tree $T$ is chosen independently from the estimators of Bargmann invariants. Combining \Cref{prop:inT} and \Cref{prop:notinT} immediately yields the following bound on the total MSE:
\begin{align}
    & \E\norm{G - \widehat{G}}_{2}^2 = \sum_{i,j}\E (\widehat{G_{ij}} - G_{ij})^2 \leq\\
    & \sum_{(i,j)\in T}\left(  \frac{1}{B_{ij}}\cdot \left(16\varepsilon^2 + \frac{k\sigma^2}{N}\right) \right) + \sum_{(i,j)\notin T} \left( \frac{2}{B_{ij}} \left( 16\varepsilon^2 + \frac{k\sigma^2}{N} \right) +  \frac{8B_{ij}}{\abs{B_{\pi(i,j)}}^2} \cdot \left(4m^2_{ij}\varepsilon^2 + \frac{k\sigma^2}{N}   \right)  \right)
\end{align}
This already gives us the formula from \Cref{th:main-MSE-simple}:
\begin{equation}
    \E\norm{G - \widehat{G}}_{2}^2 \leq \frac{A}{N}+B \varepsilon^2
\end{equation}
with:
\begin{align}
    & A = k\sigma^2\left( \sum_{(i,j)\in T}\frac{1}{B_{ij}} + 2\sum_{(i,j)\notin T}\left(\frac{1}{B_{ij}} + \frac{ 4 B_{ij}  }{\abs{B_{\pi(i,j)}}^2} \right)\right) \label{eq:B-bound}  \\
    & B=  16 \left( \sum_{(i,j)\in T}\frac{1}{B_{ij}} + 2\sum_{(i,j)\notin T}\left(\frac{1}{B_{ij}} + \frac{ B_{ij} m^2_{ij} }{\abs{B_{\pi(i,j)}}^2} \right)\right) \label{eq:A-bound}
\end{align}

If some of the Bargmann invariants in the denominators of \eqref{eq:A-bound}, \eqref{eq:B-bound} are very small, the bounds become useless. This can be mitigated by introducing a nonzero threshold $\eta >0$ and removing edges for which $B_{ij} < \eta$, which amounts to setting the desired Gram matrix entry $G_{ij}$ to zero. In this way we obtain:
\begin{equation}
    \E\norm{G - \widehat{G}}_{2}^2 \leq \frac{A_{\eta}}{N} + B_{\eta} \varepsilon^2  + n_{\eta}\cdot \eta
\end{equation}
where $A_{\eta}, B_{\eta}$ are versions of \eqref{eq:A-bound}, \eqref{eq:B-bound} with summation only over $(i,j)$ such that $B_{ij} \geq \eta$ and $n_{\eta}$ is the number of removed edges. With additional assumptions a judicious choice of $\eta \sim \frac{1}{poly(n)}$ gives an upper bound on the error that does not depend on the magnitudes of Bargmann invariants. For example, if we choose $\eta \sim \frac{\delta^2}{n^3}$ and the resulting graph admits a spanning tree which is a star, the total number of samples needed to get accuracy $\delta$ scales as $O(\frac{n^{11}}{\delta^6})$.

The above expressions are valid for any fixed spanning tree $T$ chosen independently of the estimators, but do not readily give an effective recipe for finding a $T$ which minimizes the total error since \eqref{eq:A-bound}, \eqref{eq:B-bound} are complicated functions of $T$ and corresponding Bargmann invariants $B_{\pi(i,j)}$. We propose a heuristic choice of $T$ that tries to maximize the worst denominator $B_{\pi(i,j)}$ occurring in the error bounds and can be found by a polynomial time algorithm. We describe this procedure in the following section.

\subsection{Finding the optimal spanning tree}\label{sec:tree}

Let $\Gamma$ be a connected graph on $n$ vertices with edge weights $B_{ij}$ for $(i,j) \in \Gamma$. For a spanning tree $T$ of $\Gamma$ define:
\begin{equation}
     M(T) := \min_{\pi \in T}\prod_{u_l \in \pi}B_{u_{l},u_{l+1}}
\end{equation}
i.e. the minimum product of edge weights along a path in the tree $T$. By \eqref{eq:b-notin-T} we see that:
\begin{equation}\label{eq:bargmann-lowerbound}
    \abs{B_{\pi(i,j)}} \geq \sqrt{B_{ij}} \cdot \sqrt{M(T)}
\end{equation} 
We will now show how to efficiently find $T$ which maximizes $M(T)$. By \eqref{eq:bargmann-lowerbound} this in turn provides a desired upper bound on the error terms in \eqref{eq:A-bound}, \eqref{eq:B-bound} that depend on $\frac{1}{ \abs{B_{\pi(i,j)}}^2}$.

Consider an arbitrary weighted graph with edge weights $w_{ij} \in [0,1]$. For a path $\pi =(v_1,\dots,v_k)$ connecting two vertices we define its weight as:
\begin{equation}
    w(\pi) := \prod_{i=1}^{k-1}w(u_i,u_{i+1})
\end{equation}

We want to find the spanning tree $T$ that maximizes $\min_{\pi \in T}w(\pi)$, where the maximum is taken over all paths in $T$. Note that for any path $\pi$:
\begin{equation}
        w(\pi) = \prod_{i=1}^{k-1}w(u_i,u_{i+1}) = \exp\left(\sum_{i=1}^{k-1}\log w(u_i, u_{i+1})\right) =
        \exp\left(-\sum_{i=1}^{k-1}\gamma(u_i, u_{i+1})\right)
\end{equation}
where we define $\gamma(u,v) = -\log w(u,v)$. Note that since $w(u,v) \in [0,1]$, the weights $\gamma$ are nonnegative and obviously for any spanning tree $T$:
\begin{equation}
    \min_{\pi \in T}w(\pi) = \min_{\pi \in T} \exp\left(-\sum_{i=1}^{k-1}\gamma(u_i, u_{i+1})\right) = 
    \exp\left(-\max_{\pi \in T}\sum_{i=1}^{k-1}\gamma(u_i, u_{i+1})\right)
\end{equation}
Thus, we equivalently need to find a spanning tree $T$ that minimizes $\max_{\pi \in T}\gamma(\pi)$, where:
\begin{equation}
    \gamma(\pi) := \sum_{i=1}^{k-1}\gamma(u_i, u_{i+1}).
\end{equation}
This is equivalent to solving the Minimum Diameter Spanning Tree (MDST) problem for a weighted graph, which can be easily solved in time $O(mn + n^2 \log n)$ \cite{mdst}.

By plugging \eqref{eq:bargmann-lowerbound} into \eqref{eq:A-bound}, \eqref{eq:B-bound} we obtain the following bounds:
\begin{align}
    &  A\leq  k\sigma^2\left( \sum_{(i,j)\in T}\frac{1}{B_{ij}} + 2\sum_{(i,j)\notin T}\left(\frac{1}{B_{ij}} + \frac{ 4   }{M(T)} \right)\right) \label{eq:B-bound-mt}\\
    & B \leq 16 \left( \sum_{(i,j)\in T}\frac{1}{B_{ij}} + 2\sum_{(i,j)\notin T}\left(\frac{1}{B_{ij}} + \frac{  m^2_{ij} }{M(T)} \right)\right) \label{eq:A-bound-mt} 
\end{align}

More refined choices of $T$, resulting in a smaller error bound, are possible. However, we expect that in most cases of practical interest the choice of a sufficiently good spanning tree will not be the bottleneck of the protocol (e.g. if all Bargmann invariants are close to $1$, we can simply choose $T$ to be a star rooted at arbitrary vertex), so we have not optimized the choice of the tree further.
\section{Condition number proof}\label{sec: Conditioning number}

In this section we prove that the inversion procedure from \Cref{sec:relational-info}, used to estimate Bargmann invariants $\tr(\rho_k \rho_l)$ of order two from correlators $\langle n_i n_j \rangle$ by inverting a linear system, is well conditioned for generic interferometers once $m \gg n^2$ (the diluted regime). Our proof technique falls short of proving the same assuming only $m \gg n$ (the saturated regime), but we present strong numerical evidence that the system is in fact well conditioned as soon as $m \gg n$. We also provide numerical evidence for well conditioning of linear system required to solve for Bargmann invariants of order three (see discussion in \Cref{sec:relational-info}).

The proof strategy reduces the condition number to bounding, with high probability the singular values of the linear system matrix. This is achieved using basic random matrix methods together with computing Haar measure integrals via Weingarten calculus and then using concentration inequalities over the unitary group. 

\subsection{Setup}

Recall the linear system described by \eqref{eq:Axb1}, \eqref{eq:Axb2}:
\begin{equation} \label{eq:linear-system}
    A_2(U)x_2 = b_2
\end{equation}
where $x_2 = (\tr(\rho_k \rho_l))_{k \neq l}$ is the vector of Bargmann invariants of order two, $b_2 = (b_{ij})_{i \neq j}$ is the right hand side vector of values corresponding to shifted correlators (cf. \eqref{eq:Axb1}) and $A(U)$ is the matrix depending on the interferometer:
\begin{equation}\label{eq:A(U)}
    A_2(U)_{ij;kl} = U_{ik}U_{jl}\overline{U_{il}} \overline{U_{jk}}  
\end{equation}
Since both $x_2$ and $b_2$ are symmetric with respect to changing indices $k \leftrightarrow l, i \leftrightarrow j$, we can retain only their parts indexed, respectively, by unordered pairs $e=\{k,l\}$ and $\alpha=\{i,j\}$, denoted by $x$ and $b$. The linear system \eqref{eq:linear-system} is then equivalent to:
\begin{equation}
    A(U) x = b
\end{equation}
where $A$ is an $\binom{m}{2}\times \binom{n}{2}$ matrix with entries:
\begin{equation}\label{eq:A-symmetrized}
    A(U)_{e,\alpha} := (A_2)_{ij,kl} + (\overline{A_2})_{ij,kl} = 2 \mathrm{Re}\left( U_{ik}U_{jl}\overline{U_{il}} \overline{U_{jk}} \right) 
\end{equation}
We will write $A$ instead of $A(U)$, where implicitly $U$ is taken to be a Haar random unitary matrix.

It is standard in numerical linear algebra that the relative error of the solution $x$ with respect to the error in the input $b$ depends on the condition number of the matrix $A$, equal to $\kappa(A) := \norm{A} \cdot\norm{A^{-1}}$, where $A^{-1}$ is the Moore-Penrose pseudoinverse. The system is said to be well conditioned if $\kappa(A)$ is sufficiently small -- in our setting this will mean that $\kappa(A)=O(1)$ (i.e. does not depend on the system size).

We will show that with high probability over the choice of $U$ the norms $\norm{A}=O(m^{-1}), \norm{A^{-1}}=O(m)$, which will imply that $\kappa(A)$ is of order $O(1)$. To this end, it suffices to bound the largest and smallest singular values $\sigma_{max}, \sigma_{min}$ of the matrix $A$, which are equal to the square roots of largest and smallest eigenvalues $\lambda_{max}, \lambda_{min}$ of $A^{T}A$. The bound $ \beta \leq \lambda_{min}, \lambda_{max} \leq \alpha$ is equivalent to:
\begin{equation}
    \beta I \leq A^{T}A \leq \alpha I
\end{equation}
in the positive semidefinite order.

Our main result is the following theorem:

\begin{theorem}\label{th:main-condition}
    Assume $m \geq C\cdot n^2 (\log n)^{16}$ for a sufficiently large constant $C$. Then with probability at least $1-\frac{1}{n}$ over the choice of $U$ the matrix $A$ satisfies:
    \begin{equation}
        \frac{1}{4}\cdot\frac{1}{m^2} I \leq A^{T}A \leq \frac{2}{m^2}I 
    \end{equation}
In particular, the condition number $\kappa(A)$ is with high probability bounded by $\sqrt{8} < 3$. 
\end{theorem}

\begin{figure}[thbp!]
    \centering
    \begin{minipage}[b]{0.48\textwidth}
        \centering
        \includegraphics[width=\textwidth]{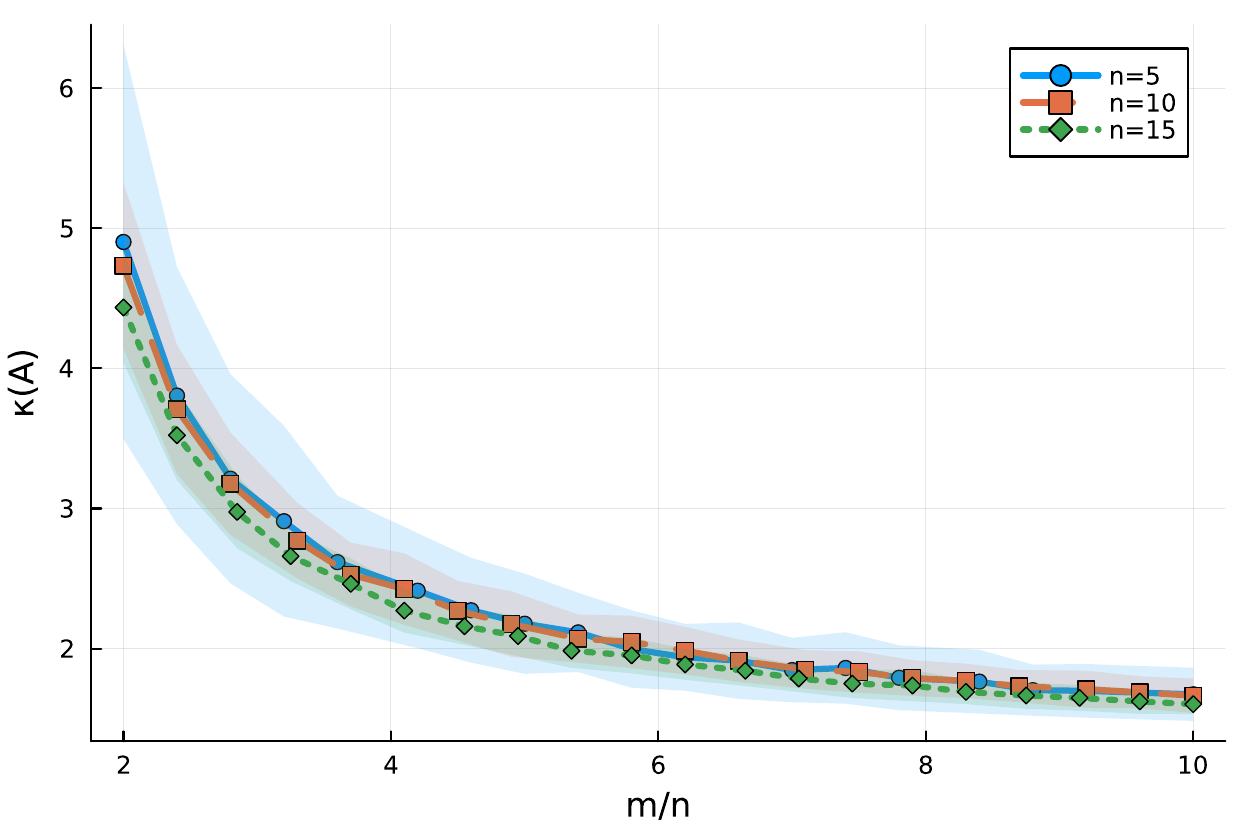}
        \caption*{(a) Condition number of $A_{2}$}
        \label{fig:figure1}
    \end{minipage}
    \hfill
    \begin{minipage}[b]{0.48\textwidth}
        \centering
        \includegraphics[width=\textwidth]{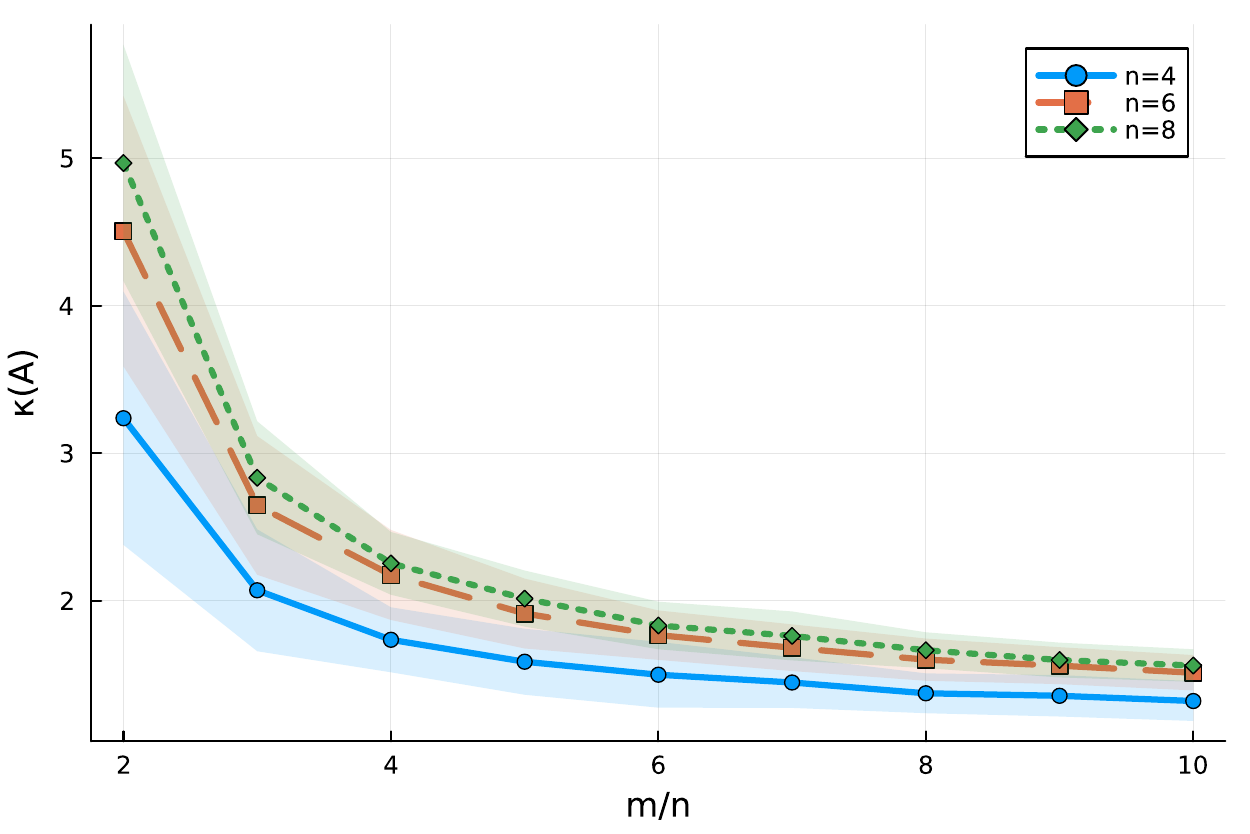}
        \caption*{(b) Condition number of $A_{3}$}
        \label{fig:figure2}
    \end{minipage}
    \caption{\textbf{Numerical conditioning of linear systems for second- and third-order Bargmann invariants.} 
    (a) Average condition number $\kappa(A_2)$ of the linear system $A_2 x_2 = b_2$ for second-order Bargmann invariants $|G_{ij}|^2$, plotted against the mode-to-photon ratio $m/n$ for $n \in \{5, 10, 15\}$. 
    (b) Average condition number $\kappa(A_3)$ of the linear system $A_3 x_3 = b_3$ for third-order Bargmann invariants $G_{ij}G_{jk}G_{ki}$, plotted against $m/n$ for $n \in \{4, 6, 8\}$. 
    Shaded regions denote one standard deviation across realizations of Haar-random unitaries $U \in U(m)$. In both cases, the condition number remains low ($\kappa \lesssim 5$) already in the regime $m = 2n$ and rapidly plateaus to $\kappa \approx 1.5$--$2.0$, confirming that correlator inversion remains numerically stable well into the dense mode regime.}
    \label{fig:conditioning_numbers}
\end{figure}

While Theorem~\ref{th:main-condition} rigorously guarantees that the linear system $A_2 x_2 = b_2$ is well-conditioned in the dilute regime ($m \gg n^2$), the analytical bound $m \gg n^2 (\log n)^{16}$ is overly conservative. In practical photonic architectures operating near the dense limit, well-conditioned inversion is expected to persist across the entire linear-scaling regime $m \ge (1+\delta)n$ for any $\delta > 0$, as suggested by Proposition \ref{prop:expectation}. 

To test this behavior, we numerically analyze the condition numbers $\kappa(A_2)$ and $\kappa(A_3)$ for reconstructing second- and third-order Bargmann invariants, respectively. For fixed photon numbers $n$ and mode counts $m$, we sample passive linear interferometers $U$ uniformly from the Haar measure on $U(m)$ to construct the measurement matrices $A_2$ and $A_3$. As shown in Fig.~\ref{fig:conditioning_numbers}, both $\kappa(A_2)$ and $\kappa(A_3)$ drop steeply as the ratio $m/n$ increases, flattening to a baseline of $\kappa \approx 1.5$--$2.0$ for $m/n \ge 5$. Down to $m/n = 2$---a lower bound chosen arbitrarily for numerical display rather than a physical transition---the condition numbers remain consistently small ($\kappa \lesssim 5$) with tight concentration across random unitaries. Crucially, the profiles for different values of $n$ collapse onto a single universal curve when plotted against $m/n$, exhibiting scale-invariance with system size. These results confirm that the quadratic requirement $m \gg n^2$ is an artifact of the proof technique, and generic passive interferometers enable stable, full Gram-matrix recovery whenever $m$ scales linearly with $n$.

\subsection{Expectation and variance}

Let us decompose $A^{T}A$ as:
\begin{equation}
    A^{T}A = \E A^{T}A + (A^{T}A - \E A^{T}A)
\end{equation}
where the expectation is taken with respect to the Haar random interferometer $U$. We will first compute the expectation $\E A^{T}A$ and bound its eigenvalues. Next, we will prove a high probability bound for the operator norm $\norm{A^{T}A - \E A^{T}A}_{\infty}$ of the error matrix. These two will overall give us a high probability bounds on $\lambda_{max}$ and $\lambda_{min}$ of $A^{T}A$. This strategy is rather crude since it uses just bounding the magnitude of each matrix entry to produce bounds on eigenvalues. We expect that a more sophisticated approach could lead to relaxing the condition in \Cref{th:main-condition} to $m\gg n$, but would face more mathematical difficulties.

We first write the matrix entries explicitly. Writing:
\begin{equation}
    p(i,j,k,l) := U_{ik}U_{jl}\overline{U_{il}} \overline{U_{jk}}  
\end{equation}
by \eqref{eq:A-symmetrized} we have for $\alpha=\{k,l\},\alpha'=\{k',l'\}$:
\begin{equation}\label{eq:AA(U)}
    (A^{T}A)_{\alpha,\alpha'} = \sum_e A_{e,\alpha}A_{e,\alpha'} =
    \sum_{i<j}  (p(i,j,k,l) + p(i,j,l,k))(p(i,j,k',l') + p(i,j,l',k'))
\end{equation}

We will need to compute the expectation and variance of each entry with respect to the choice of the random interferometer $U$, which are supplied by the following lemmas. Since each entry of \eqref{eq:AA(U)} is a polynomial in the entries of $U$, such averages can be computed by the use of Weingarten calculus. We describe this method together with proofs of the lemmas in \Cref{sec:weingarten}.

\begin{lemma}\label{lm:expectation-weingarten}
We have:    
\begin{align}
 \E(A^{T}A)_{\alpha,\alpha'} =\frac{1}{D_m}\cdot
    \begin{cases}
        m^2 + m + 10, & \text{$\alpha=\alpha'$} \\
        -(m-5), & \text{$\vert \alpha \cap \alpha' \vert=1$} \\
        4, & \text{$\alpha \cap \alpha' =\emptyset$}
    \end{cases}
\end{align}
where $D_m = m(m+1)(m+2)(m+3)$
\end{lemma}

\begin{lemma}\label{lm:variance-weingarten}
For $m\geq 8$ we have:
\begin{align}
 \mathrm{Var}(A^TA)_{\alpha,\alpha'} \leq
    \begin{cases}
        \frac{4}{m^5}, & \text{$\alpha=\alpha'$} \\
        \frac{5}{m^6}, & \text{otherwise}
    \end{cases}
\end{align}
\end{lemma}

\begin{proposition}\label{prop:expectation}
    Assume that $m \geq (1+\delta)n$ for some constant $\delta >  0$. Then the expected matrix $\E A^{T}A$ satisfies:
    \begin{equation}
        \frac{\delta}{1+\delta}\cdot\frac{1}{m^2} I \leq \E A^{T}A \leq \frac{1}{m^2} I
    \end{equation}
\end{proposition}

\begin{proof}
Let $J_1$ be the matrix given by $(J_1)_{\alpha,\alpha'} := \id_{\{\vert \alpha \cap \alpha' \vert=1\}}$ and likewise $J_2$ given by $(J_2)_{\alpha,\alpha'} := \id_{\{\alpha \cap \alpha' =\emptyset\}}$. Note that $J_2 = \id \id^T - I - J_1$, where $\id$ denotes the all-ones vector. By \Cref{lm:expectation-weingarten} we then have:
\begin{equation}
    \E A^TA = \frac{1}{D_m}\left( (m^2+m+10)I - (m-5)J_1 + 4 J_2 \right) =
    \frac{1}{D_m}\left( 4 \id \id^T + (m^2+m+6)I - (m-1)J_1 \right)
\end{equation}
where $D_m = m(m+1)(m+2)(m+3)$. The all-ones matrix $\id \id^T$ has $\id$ as the eigenvector with eigenvalue $\binom{n}{2}$ and is zero elsewhere. The matrix $J_1$ is the adjacency matrix of the Johnson graph $J(n,2)$. Its eigenvalues are $(-2, n-4, 2n-4)$, and the vector $\id$ has eigenvalue $2n-4$. Thus the matrix can be diagonalized exactly, with eigenvalues:
\begin{align}
    &\lambda_1 = \frac{m^2 + (5-2n)m + 2n^2 + 2}{D_m} \\ 
    &\lambda_2 = \frac{m^2 + (5-n)m + n + 2}{D_m} \\ 
    &\lambda_3 = \frac{m^2 + 3m+4}{D_m} 
\end{align}
Using the assumption $n \leq \frac{m}{1+\delta}$:
\begin{itemize}
    \item For $\lambda_3$, we have $\lambda_3 = \frac{m^2+3m+4}{m(m+1)(m+2)(m+3)} \leq \frac{1}{m^2}$ and $\lambda_3 \geq \frac{1}{m^2+3m+2} \geq \frac{\delta}{1+\delta}\frac{1}{m^2}$.
    \item For $\lambda_2$, the numerator is decreasing in $n$, so substituting $n \leq \frac{m}{1+\delta}$ gives:
    \begin{equation}
        m^2 + (5-n)m + n + 2 \geq m^2\left(1 - \frac{1}{1+\delta}\right) + m\left(5 + \frac{1}{1+\delta}\right) + 2 = \frac{\delta}{1+\delta}m^2 + \left(5 + \frac{1}{1+\delta}\right)m + 2
    \end{equation}
    which yields $\lambda_2 \geq \frac{\delta}{1+\delta}\frac{1}{m^2}$, with the maximum attained at $n=2$ satisfying $\lambda_2 \leq \lambda_3 \leq \frac{1}{m^2}$.
    \item For $\lambda_1$, the numerator $2(n - m/2)^2 + \frac{1}{2}m^2 + 5m + 2$ is minimized at $n = m/2 \leq \frac{m}{1+\delta}$ (since $\delta \leq 1$), giving a lower bound of $\frac{1}{2m^2} \geq \frac{\delta}{1+\delta}\frac{1}{m^2}$, and an upper bound bounded by $\frac{1}{m^2}$.
\end{itemize}
Hence, the smallest eigenvalue is $\geq \frac{\delta}{1+\delta}\frac{1}{m^2}$ and the largest is $\leq \frac{1}{m^2}$, proving the result.
\end{proof}
\begin{remark}
    Notice that from the previous proof we can also deduce that the ratio between the largest and smallest eigenvalue of $A^{T}A$ in the large $m$ limits behaves as
    \begin{equation}
        \left.\frac{\lambda_{\max}}{\lambda_{\min}}\right|_{m\gg 1}\geq  \begin{cases}
            1+\frac{1}{\delta} & \text{for} \ \delta\in (0,1] \\
            1+\frac{2\delta}{1+\delta^{2}} & \text{for} \ \delta>1
        \end{cases}
    \end{equation}
    which clearly diverges in the limit $\delta\to 0^{+}$ and converges to $1$ in the limit $\delta\to +\infty$.
\end{remark}
We will now need to bound the error term $A^{T}A - \E A^{T}A$. The following lemma says that if the error term is sufficiently small entrywise, then \Cref{th:main-condition} holds.


\begin{lemma}\label{lm:main-lemma}
Let $M = A^{T}A - \E A^{T}A $. Suppose that for each $\alpha$ we have $\vert M_{\alpha,\alpha} \vert\leq \frac{1}{8}\cdot \frac{1}{m^2}$ and for each $\alpha \neq \alpha'$ we have $\vert M_{\alpha,\alpha'} \vert \leq \frac{1}{8n^2}\cdot\frac{1}{m^2}$. Then: 
\begin{equation}
        \frac{1}{4}\cdot \frac{1}{m^2} I\leq A^{T}A\leq \frac{2}{m^2}I 
    \end{equation}
\end{lemma}

\begin{proof}
    We make use of the easily verified fact: if an $l\times l$ matrix $O$ satisfies $\abs{O_{ij}} \leq K$ for each $i,j$, then $\norm{O}_{\infty} \leq Kl$. Writing $M = D + O$, where $D$ is the diagonal part and $O$ the off-diagonal part, we have:
    \begin{equation}
        \norm{M}_{\infty} \leq \norm{D}_{\infty} + \norm{O}_{\infty} \leq \max_{\alpha} \vert M_{\alpha,\alpha} \vert + \binom{n}{2} \max_{\alpha \neq \alpha'}\vert M_{\alpha,\alpha'} \vert \leq \frac{1}{8}\frac{1}{m^2} + n^2 \cdot \frac{1}{8n^2}\frac{1}{m^2} < \frac{1}{4m^2}
    \end{equation}
    The claim then follows by invoking  \Cref{prop:expectation} and noting that $\lambda_{min}(A+B) \geq \lambda_{min}(A) - \norm{B}_{\infty}, \lambda_{max}(A+B) \leq \lambda_{max}(A) + \norm{B}_{\infty}$.
\end{proof}

To prove \Cref{th:main-condition} we are thus left with ensuring that conditions of \Cref{lm:main-lemma} are satisfied with high probability -- that is, each entry of the error matrix $M$ is small. This is done in the next section using concentration of measure for polynomials.

\subsection{Concentration for unitary polynomials}

To show that the conditions of \Cref{lm:main-lemma} hold with high probability requires proving a concentration bound on the entries of $A^{T}A$. In principle one can obtain such a bound by computing sufficiently high moments -- however, the size of corresopnding Weingarten calculus computations quickly become prohibitive. Instead, we will prove concentration with tails scaling like $\exp(-c t^{1/16})$ using the connection between log-Sobolev constants, hypercontractivity and bounds on eigenvalues of the Laplace-Beltrami operator on the unitary group. The ingredients used are standard and well known, but we have not seen them applied specifically to obtain a concentration bound for polynomials on the unitary group. The claimed concentration bound will have the same form as the one for Gaussian polynomials, see e.g. \cite{abmb}, Corollary 5.49.

We start with the observation that the unitary group $U(m)$ satisfies the log-Sobolev inequality for every locally Lipschitz function $f$:
\begin{equation}
    \mathrm{Ent}(f^2) \leq 2C_{LS} \int \E\abs{\nabla f}^2
\end{equation}
where $\mathrm{Ent}(g) = \E g \log (g) - (\E g) (\log\E g)$, with log-Sobolev constant $C_{LS} = \frac{6}{m}$ \cite{meckes}.

As is well known, log-Sobolev constant implies hypercontractivity. More precisely, let $\Delta$ be the Laplace-Beltrami operator on $U(m)$ and let $P_t = e^{t\Delta}$ be the corresponding heat semigroup. Fix $q > p > 1$. Then for any $t \geq \frac{C_{LS}}{2}\log\frac{q-1}{p-1}$ and $f \in L^p$ we have the hypercontractive bound:
\begin{equation}\label{eq:hypercontractivity}
    \norm{P_t f}_{q} \leq \norm{f}_{p}
\end{equation}

The bound \eqref{eq:hypercontractivity} can be used to derive concentration bounds for polynomials if one has bounds on the eigenvalues of $\Delta$ restricted to the space of polynomials of degree $(k,k)$. We have the following bound on eigenvalues:

\begin{lemma}\label{lm:eigenvalues}
    Let $\lambda$ be the label of any irrep of $U(m)$ that appears in the decomposition of the representation $U^{\otimes k} \otimes \overline{U^{\otimes k}}$, i.e. the space of polynomials of degree $(k,k)$. Assume $k \leq m$. Let $f$ be any matrix element of corresponding to $\lambda$. Then:
    \begin{equation}
        \Delta f = - k_{\lambda}f
    \end{equation}
    where $\abs{k_{\lambda}} \leq 4km$.
\end{lemma}

\begin{proof}
    By standard analysis on Lie groups (see e.g. \cite{analysis}, Section 12.1) the Laplacian $\Delta$ is equivalent to the Casimir operator, which acts on the irrep space $\lambda$ as $k_{\lambda}\cdot Id$, where:
    \begin{equation}
        k_{\lambda} = \sum_{i=1}^{m} (\lambda_i^2 + (m - 2i + 1)\lambda_i)
    \end{equation}
    The irreps $\lambda$ that appear in the decomposition of the representation $U^{\otimes k} \otimes \overline{U^{\otimes k}}$ correspond to weight vectors $\lambda=(\lambda_1,\dots,\lambda_m)$ such that $\sum_i \lambda_i =0$ and $\sum_i\abs{\lambda_i} \leq 2k$. Let:
    \begin{equation}
        r := \sum_{\lambda_i>0}\lambda_i = \frac{1}{2}\sum_i \vert \lambda_i \vert
    \end{equation}
    By assumption $r \leq k$. Clearly $\sum_i \lambda_i^2 \leq 2r^2$. Moreover, also $\sum_i \vert m+1-2i \vert\cdot\vert\lambda_i \vert\leq 2r (m-1)$ because $\vert\lambda_i\vert \leq 2r$. Hence by $k\leq m$ we obtain:
    \begin{equation}
        \abs{k_{\lambda}} \leq 2r^2 + 2(m-1)r \leq 2k^2 + 2mk \leq 4km
    \end{equation}
\end{proof}

\begin{lemma}\label{lm:hypercontractivity-poly}
    Let $Q$ be a polynomial of degree $(k,k)$. Assume $m \geq k$. Then for every $q \geq 2$:
    \begin{equation}
        \norm{Q}_{q} \leq (q-1)^{4k} \norm{Q}_{2}
    \end{equation}
\end{lemma}

\begin{proof}
Let $\Lambda_k$ denote the set of irreps $\lambda$ of $U(m)$ that appear in the decomposition of the representation $U^{\otimes k} \otimes \overline{U^{\otimes k}}$, i.e. the space of polynomials of degree $(k,k)$. Let us decompose $Q$ into orthonormal eigenfunctions $f_{\lambda,i}$, where $\lambda$ labels irreps and $i$ is the index within the irrep:
\begin{equation}
    Q = \sum_{\lambda,i}c_{\lambda,i}f_{\lambda_i}
\end{equation}
where $\sum_{\lambda,i}\abs{c_{\lambda_i}}^2=\norm{Q}_{2}^{2}$. 

First, note that each $Q$ is invariant with respect to $Q(U) \rightarrow Q(e^{i\theta}U)$ because the degree is balanced. Thus each $Q$ determines a function on $SU(m)$ and we can use the better log-Sobolev constant for $SU(m)$, which is equal to $\frac{2}{m}$ \cite{meckes}. Consequently, choose $t = \frac{1}{m}\log(q-1)$. Since $Q = P_t P_{-t}Q$, by the hypercontractive bound \eqref{eq:hypercontractivity} we get:
\begin{equation}
    \norm{Q}_{q} \leq \norm{P_{-t}Q}_2
\end{equation}
By \Cref{lm:eigenvalues}, each $f_{\lambda,i}$ is an eigenfunction of $P_{-t}$ with eigenvalue $e^{k_{\lambda}t}$, where $k_{\lambda} \leq 4km$. Since $t = \frac{1}{m}\log(q-1)$, we have $e^{mt} = q-1$. Therefore:
\begin{equation}
    \norm{P_{-t}Q}_2^2 = \sum_{\lambda,i}e^{2k_{\lambda}t}\abs{c_{\lambda,i}}^2 \leq e^{8kmt}\norm{Q}_{2}^{2} = (q-1)^{8k}\norm{Q}_{2}^{2}
\end{equation}
\end{proof}

We are now ready to prove the main concentration bound:
\begin{theorem}\label{th:polynomial}
    Let $Q$ be a polynomial on $U(m)$ of degree $(k,k)$. Assume $m\geq k$. Then for any $t > (2e)^{4k}$:
    \begin{equation}\label{eq:concentration-bound}
        \mathbb{P}\left(\vert Q -\E Q\vert \geq t \sqrt{\mathrm{Var}(Q)}\right) \leq \exp\left(- \frac{k}{e}\cdot t^{1/4k}\right)
    \end{equation}
\end{theorem}

\begin{proof}
By rescaling we can assume $\E Q = 0, \mathrm{Var}(Q)=1$. By Markov's inequality and \Cref{lm:hypercontractivity-poly} for any $q\geq 2$:
\begin{equation}
    \mathbb{P}\left( \abs{Q} \geq t \right) \leq \frac{\E\abs{Q}^q}{t^q} \leq \left(\frac{(q-1)^{4k}}{t}\right)^q \leq \left(\frac{q^{4k}}{t}\right)^q
\end{equation}
One easily verifies that the choice $q = \frac{1}{e}\cdot t^{1/4k}$ yields the result, provided that $t > (2e)^{4k}$ so that $q > 2$.
\end{proof}

\begin{remark}
    The above argument gives the non-optimal tail exponent $\frac{1}{4k}$ for $m \geq k$. With some additional work this can be improved to $\frac{1}{2k}$. The optimal exponent $\frac{1}{k}$ (with different, non-explicit constants) can be obtained using recent results about hypercontractivity for Lie groups \cite{productmixing}. We retain the current proof due to its relatively elementary nature and explicit constants.
\end{remark}

We can now combine all the ingredients to prove \Cref{th:main-condition}.

\begin{proof}[Proof of \Cref{th:main-condition}]
We will prove that with probability at least $1-\frac{1}{n}$ the conditions of \Cref{lm:main-lemma} are satisfied. It suffices to prove that: 1) with probability $\geq 1 - \frac{1}{2n}$ for all diagonal terms we have $\vert M_{\alpha,\alpha} \vert\leq \frac{1}{8}\cdot \frac{1}{m^2}$; 2) with probability $\geq 1 - \frac{1}{2n}$ for each $\alpha \neq \alpha'$ we have $\vert M_{\alpha,\alpha'} \vert \leq \frac{1}{8n^2}\cdot\frac{1}{m^2}$. We first handle the diagonal terms. By union bounding over the $\binom{n}{2}$ possible choices of $\alpha$ it suffices to prove that:
\begin{equation}
\mathbb{P}\left( \vert M_{\alpha,\alpha} \vert\geq \frac{1}{8}\cdot\frac{1}{m^2}\right) \leq \frac{1}{n^3}.    
\end{equation}
We invoke \Cref{th:polynomial} applied to the polynomial $Q = (A^TA)_{\alpha,\alpha}$. Recall that by \Cref{lm:variance-weingarten} the variance of each diagonal term is bounded by $\frac{4}{m^{5}}$. We need a value of $t$ in \Cref{th:polynomial} such that the probability in \eqref{eq:concentration-bound} is at most $\frac{1}{n^3}$ and at the same time $t \sqrt{\mathrm{Var}((A^TA)_{\alpha,\alpha})} < \frac{1}{8}\cdot \frac{1}{m^2}$. After a short computation, we see that such a $t$ exists if:
\begin{equation}
    \frac{1}{16}\sqrt{m} > (3 \log n)^{16},
\end{equation}
which is satisfied by the assumption that $m \geq C n^2 (\log n)^8$ for a sufficiently large constant $C$. The off-diagonal case is handled in the same way, only now we union bound over $\binom{n}{2}^2 - \binom{n}{2}$ off-diagonal entries and thus require the relevant probability for each entry to be $< \frac{1}{n^5}$, and we require that $t \sqrt{\mathrm{Var}((A^TA)_{\alpha,\alpha'})} < \frac{1}{8n^2}\cdot \frac{1}{m^2}$. Again by \Cref{lm:variance-weingarten} this leads to:
\begin{equation}
  \frac{1}{8\sqrt{5}}\cdot\frac{m}{n^2} > (6 \log n)^{16},  
\end{equation}
which is again satisfied if $m \geq C n^2 (\log n)^{16}$.
\end{proof}

\subsection{Weingarten calculus}\label{sec:weingarten}

We recall the basic unitary Weingarten formula. For \(d\geq 1\) and multi-indices
\(\mathbf{i}=(i_1,\ldots,i_d)\), \(\mathbf{j}=(j_1,\ldots,j_d)\), let
\(U_{\mathbf{i},\mathbf{j}}:=\prod_{a=1}^d U_{i_a j_a}\). For
\(\sigma\in S_d\) put
\[
\delta_\sigma(\mathbf{i},\mathbf{i}')
:=
\prod_{a=1}^d \delta_{i_a,i'_{\sigma(a)}}.
\]
Then, for a Haar random unitary matrix \(U\in U(m)\) \cite{Collins:2006jgn}:
\begin{equation}
\mathbb{E}\!\left[
U_{\mathbf{i},\mathbf{j}}
\overline{U_{\mathbf{i}',\mathbf{j}'}}
\right]
=
\sum_{\sigma,\tau\in S_d}
\delta_\sigma(\mathbf{i},\mathbf{i}')
\delta_\tau(\mathbf{j},\mathbf{j}')
\operatorname{Wg}_m(\sigma^{-1}\tau),
\label{eq:unitary-weingarten}
\end{equation}
where \(\operatorname{Wg}_m\) is the unitary Weingarten function. For \(m\geq d\) it can be written as \cite{Collins:2006jgn}:
\begin{equation}
\operatorname{Wg}_m(\pi)
=
\frac{1}{d!}
\sum_{\lambda\vdash d}
\frac{f^\lambda\chi^\lambda(\pi)}
{\displaystyle\prod_{(a,b)\in\lambda}(m+b-a)}.
\label{eq:weingarten-character}
\end{equation}
Here \(\chi^\lambda\) is the irreducible character of \(S_d\) corresponding to the partition
\(\lambda\) and \(f^\lambda=\chi^\lambda(\mathrm{id})\). In particular,
\(\operatorname{Wg}_m(\pi)\) depends only on the cycle type of \(\pi\).

\begin{proof}[Proof of \Cref{lm:expectation-weingarten}]
Fix \(e=\{i,j\}\). Expanding \(A_{e,\alpha}A_{e,\alpha'}\) using \eqref{eq:A-symmetrized} gives four monomials of degree \((4,4)\).
Applying \eqref{eq:unitary-weingarten} and \eqref{eq:weingarten-character}, we can compute numerically how many times each cycle type occurs and the corresponding Weingarten function.

In the end we obtain:
\begin{equation}
\mathbb{E}\!\left[A_{e,\alpha}A_{e,\alpha'}\right]
=
\frac{2}{m^2(m-1)(m+1)(m+2)(m+3)}
\begin{cases}
m^2+m+10, & \alpha=\alpha',\\
-(m-5), & |\alpha\cap\alpha'|=1,\\
4, &  \alpha\cap\alpha'=\varnothing.
\end{cases}
\label{eq:single-row-pair-expectation}
\end{equation}
The right hand side does not depend on \(e\), and there are \(\binom{m}{2}\) possible
choices of \(e\). Therefore:
\[
\mathbb{E}(A^T A)_{\alpha,\alpha'}
=
\sum_e\mathbb{E}\!\left[A_{e,\alpha}A_{e,\alpha'}\right]
=
\frac{1}{m(m+1)(m+2)(m+3)}
\begin{cases}
m^2+m+10, & \alpha=\alpha',\\
-(m-5), & |\alpha\cap\alpha'|=1,\\
4, & \alpha\cap\alpha'=\varnothing,
\end{cases}
\]
which is exactly what is required in \Cref{lm:expectation-weingarten}.
\end{proof}

\begin{proof}[Proof of \Cref{lm:variance-weingarten}]
Fix \(\alpha,\alpha'\) and put
\[
X_e:=A_{e,\alpha}A_{e,\alpha'},
\qquad
(A^T A)_{\alpha,\alpha'}=\sum_e X_e.
\]
By invariance of Haar measure under permutations of rows and columns, the quantity
\(\mathbb{E}X_eX_f\) depends only on
\[
r:=|e\cap f|,
\qquad
s:=|\alpha\cap\alpha'|.
\]
We denote this quantity by \(M_{rs}\). Expanding \(X_eX_f\) gives sixteen monomials of 
degree \((8,8)\). We again apply \eqref{eq:unitary-weingarten} and \eqref{eq:weingarten-character}  with \(d=8\), computing numerically the relevant cycle types and  Weingarten functions. Let
\[
Q_m
:=
m^2(m+1)^2(m+2)^2(m+3)(m+4)(m+5)(m+6)(m+7).
\]
Then:
\begin{equation}
\begin{array}{c|ccc}
Q_m M_{rs} & s=2 & s=1 & s=0\\
\hline
r=2
&
\dfrac{96P_{22}}{m-1}
&
\dfrac{16P_{21}}{m-1}
&
\dfrac{4P_{20}}{m-1}
\\[2mm]
r=1
&
\dfrac{16P_{21}}{m-1}
&
-\dfrac{4P_{11}}{(m-1)^2}
&
\dfrac{8P_{10}}{(m-2)(m-1)^2}
\\[2mm]
r=0
&
\dfrac{4P_{20}}{m-1}
&
\dfrac{8P_{10}}{(m-2)(m-1)^2}
&
\dfrac{8P_{00}}{(m-3)(m-2)^2(m-1)^2}
\end{array}
\label{eq:weingarten-order-eight-table}
\end{equation}
where:
\begin{align}
P_{22}
&=
m^4+10m^3+113m^2+440m+876,
&
P_{21}
&=
m^4+12m^3+107m^2+252m+540,
\nonumber\\
P_{20}
&=
m^4+14m^3+99m^2+230m+696,
&
P_{11}
&=
m^4+239m^2+324m+444,
\nonumber\\
P_{10}
&=
m^4+2m^3+175m^2+182m+200,
&
P_{00}
&=
19m^4-166m^3+857m^2+994m+648.
\label{eq:weingarten-order-eight-polynomials}
\end{align}
We now sum over the possible relative positions of \(e\) and \(f\). The numbers of ordered pairs \((e,f)\) satisfying \(|e\cap f|=k\) are $N_k = \binom{m}{2}\binom{m-k}{k}$. Let \(\mu_s:=\mathbb{E}X_e\), whose value is given by
\eqref{eq:single-row-pair-expectation}. It follows that:
\begin{equation}
\operatorname{Var}(A^T A)_{\alpha,\alpha'}
=
N_2M_{2s}+N_1M_{1s}+N_0M_{0s}
-
\binom{m}{2}^{\!2}\mu_s^2.
\label{eq:variance-from-weingarten}
\end{equation}
Substituting \eqref{eq:weingarten-order-eight-table} and simplifying, define
\[
H_m
:=
(m+1)^2(m+2)^2(m+3)^2(m+4)(m+5)(m+6)(m+7).
\]
We obtain:
\begin{align}
\operatorname{Var}(A^T A)_{\alpha,\alpha}
=
\begin{cases}
\frac{2S_2(m)}{m^2H_m}, & \alpha=\alpha'\\
\frac{S_1(m)}{m^2(m-1)H_m}, & \vert \alpha \cap \alpha' \vert=1\\
\frac{2S_0(m)}{m^2(m-2)(m-1)H_m}, & \alpha \cap \alpha' = \emptyset
\end{cases}
\label{eq:exact-entrywise-variances}
\end{align}
where:
\begin{align}
S_2(m)
&=
2m^7+37m^6+439m^5+2319m^4+8627m^3
+19716m^2+3116m-42000,
\nonumber\\
S_1(m)
&=
5m^7+101m^6+452m^5+2102m^4+6161m^3
-3299m^2-19354m+21000,
\nonumber\\
S_0(m)
&=
m^8+18m^7+115m^6+561m^5+906m^4-1283m^3
+2298m^2+11272m-13440.
\label{eq:variance-polynomials}
\end{align}
For \(m\geq 8\), a direct comparison of the numerators and denominators gives:
\[
V_2\leq \frac{4}{m^5},
\qquad
V_1\leq \frac{5}{m^6},
\qquad
V_0\leq \frac{5}{m^6}.
\]
Indeed, after clearing the positive denominators, the numerator of
\(4m^{-5}-V_2\) has positive coefficients as a polynomial in \(m\), while the
numerators of \(5m^{-6}-V_1\) and \(5m^{-6}-V_0\) have positive coefficients after
writing \(m=8+t\). This proves \Cref{lm:variance-weingarten}.
\end{proof}

\section{Analytical expression for the correlators}\label{sec: Correlator derivation}

To derive the particle numbers correlators we make use of the same approach described in \cite{anguitaExperimentalValidationBoson2025,robbio2026complementaritybosonicfermionicmanybody}, by introducing the moment generating function defined as
\begin{equation}
    \mathcal{M}(\boldsymbol{y})=\perm\left(V^{\dagger}e^{Y}V\odot G\right)
\end{equation}
where $Y=\operatorname{diag}(y_{1},...,y_{m})$ and $V\in \mathbb{C}^{m \times n}$ made by the first $n$ column of $U$, such that $V^{\dagger}e^{Y}V\in \mathbb{C}^{n\times n}$. The moment generating function satisfies the property
\begin{equation}
    \frac{\partial^{|\boldsymbol{k}|}}{\partial y_{1}^{k_{1}}\dots \partial y_{m}^{k_{m}}} \mathcal{M}(\boldsymbol{y}){\bigg |}_{\boldsymbol{y}=0}=\langle \hat{n}_{1}^{k_{1}}\dots\hat{n}_{m}^{k_{m}}\rangle
\end{equation}
where $\boldsymbol{k}=(k_{1},...,k_{m})\in \mathbb{N}^{m}$ and $|\boldsymbol{k}|=\sum_{i}k_{i}$. To make calculation simpler we rewrite the above as
\begin{align}
    \mathcal{M}(\boldsymbol{y})&=\perm\left(I-I+V^{\dagger}e^{Y}V\odot G\right)\\
    &=\perm\left(I-V^{\dagger}IV\odot G+V^{\dagger}e^{Y}V\odot G\right)\\
    &=\perm\left(I+V^{\dagger}(e^{Y}-I)V\odot G\right) \quad .
\end{align}
This rewriting allows us to use two different properties. First we recall the following result ~\cite{minc1984permanents,grinberg2022principal} for any matrix $\Gamma\in \mathbb{C}^{m\times m}$:
\begin{equation}\label{eq: perm expansion}
    \perm(I+\Gamma)=\sum_{\boldsymbol{j}\in \{0,1\}^{m}}\perm(\Gamma^{(\boldsymbol{j},\boldsymbol{j})}) \quad .
\end{equation}
Then we make use of the fact that the function $f(y)=e^{y}-1$ satisfies
\begin{equation}
    \frac{d^{k}}{dy^{k}}(e^{y}-1)^{l}=l!S(k,l),
\end{equation}
where the $S(k,l)$ are the so-called Stirling numbers of the second kind~\cite{OEISA008277}, which satisfy $S(k,l)=0$ if and only if $l>k$. This means that if we are interested in the $k$-th derivatives of the functions $\perm(I+\Gamma)$ in the calculation of $k$-th moments, we can limit ourselves to terms of order up to $k$ in the corresponding series. So for the following we define
\begin{equation}
    \Gamma_{ij}=G_{ij} \sum_{k=1}^{m} U_{ki}^* U_{kj} (e^{y_{k}}-1)
\end{equation}
and then proceed with the direct calculations. We obtain
\begin{align}\label{eq:correlators2-appendix}
    \langle \hat{n}_{i}\hat{n}_{j}\rangle &= \delta_{ij} \sum_{k=1}^n |U_{ik}|^2 + \sum_{k \neq l} \left( |U_{ik}|^2 |U_{jl}|^2 + |G_{kl}|^2 U_{ik}^* U_{il} U_{jl}^* U_{jk} \right) \quad .
\end{align}
The third order is more complicated to derive, but follows the same approach. The final result is the following:
\begin{equation}
\begin{aligned}
\langle \hat{n}_{a}\hat{n}_{b}\hat{n}_{c}\rangle
&=\sum_{r,s,t \ \textit{distinct}}^{n}
\Big[
|U_{r a}|^2\,|U_{s b}|^2\,|U_{t c}|^2+\,|U_{r a}|^2\;U_{s b}\,\overline{U_{t b}}\;U_{t c}\,\overline{U_{s c}}\;G_{st}G_{ts} \\
&+\,U_{r a}\,\overline{U_{s a}}\;U_{s b}\,\overline{U_{r b}}\;|U_{t c}|^2\;G_{rs}G_{sr} 
 +\,U_{r a}\,\overline{U_{t a}}\;|U_{s b}|^2\;U_{t c}\,\overline{U_{r c}}\;G_{rt}G_{tr} \\
 &+\,U_{r a}\,\overline{U_{s a}}\;U_{s b}\,\overline{U_{t b}}\;U_{t c}\,\overline{U_{r c}}\;G_{sr}G_{rt}G_{ts}+\,U_{r a}\,\overline{U_{t a}}\;U_{s b}\,\overline{U_{r b}}\;U_{t c}\,\overline{U_{s c}}\;G_{st}G_{tr}G_{rs}\Big].
\end{aligned}
\end{equation}

\begin{equation}
\begin{aligned}
\langle \hat{n}_{a}^{2}\hat{n}_{b}\rangle&=\sum_{r\neq s}^{n}
\Big[
|U_{r a}|^2|U_{s b}|^2
+
U_{r a}\overline{U_{s a}}\overline{U_{r b}}U_{s b}G_{rs}G_{sr}
\Big] \\
&\quad +\sum_{r,s,t \ \textit{distinct}}^{n}
\Big[
|U_{r a}|^2\,|U_{s a}|^2\,|U_{t b}|^2 
+|U_{r a}|^2\;U_{s a}\,\overline{U_{t a}}\;U_{t b}\,\overline{U_{s b}}\;G_{st}G_{ts} \\
&\qquad +\,|U_{r a}|^2\,|U_{s a}|^2\,|U_{t b}|^2\;G_{rs}G_{sr}
+\,|U_{s a}|^2\;U_{r a}\,\overline{U_{t a}}\;U_{t b}\,\overline{U_{r b}}\;G_{sr}G_{rt}G_{ts} \\
&\qquad +\,|U_{r a}|^2\;U_{s a}\,\overline{U_{t a}}\;U_{t b}\,\overline{U_{s b}}\;G_{tr}G_{rs}G_{st}
+\,|U_{s a}|^2\;U_{r a}\,\overline{U_{t a}}\;U_{t b}\,\overline{U_{r b}}\;G_{rt}G_{tr}
\Big].
\end{aligned}
\end{equation}

\begin{equation}
\begin{aligned}
\langle \hat{n}_{a}^{3}\rangle&=
\sum_{r,s,t \ \textit{distinct}}^{n}
|U_{r a}|^2\,|U_{s a}|^2\,|U_{t a}|^2 \Big[ 1+G_{st}G_{ts} +
G_{rs}G_{sr} + G_{rt}G_{tr} +
G_{ts}G_{sr}G_{rt}+ 
G_{rt}G_{ts}G_{sr}
\Big] \\
&\quad +\,3\sum_{r,s \ \textit{distinct}}^{n}
|U_{r a}|^2\,|U_{s a}|^2
\;\Big[1+G_{rs}G_{sr} 
\Big]
+\,\sum_{r=1}^{n}|U_{r a}|^2.
\end{aligned}
\end{equation}

Notice that, as discussed in the main text, the correlator of order $k$ is expressible in terms of to Bargmann invariants of order $k$ or lower. The above relation can be generalized to a product of mixed states just by substituting $G_{ij}G_{ji}$ with $\Tr{\rho_{i}\rho_{j}}$ and similarly for the third order.
\section{Certification of Indistinguishability}\label{Sec: Certification}

The quantification of indistinguishability is a task that has sparked significant interest in the context of photonic quantum computing since the majority of protocols assume indistinguishable particles, and thus understanding how indistinguishable a set of photons is, allows one to quantify the deviation from ideal behavior. In recent years many different measures of it have been proposed, including in Refs.~\cite{sanzExponentialImprovementBenchmarking2026,Brod_GI,schadow2026certification}, but here we focus on the one discussed in \cite{englbrechtIndistinguishabilityIdenticalBosons2024,OptimalInterferometricCertification}. In these references, it was shown that given a state $\Omega_{ext}(\rho)$, the fidelity to the closest indistinguishable state can be written as 
\begin{equation}
    F_{ind}(\rho)=\Tr{\Pi_{\mathrm{sym}}\Omegaext(\rho)} 
\end{equation}
and similarly if $\Omegaext(\psi_{ind})$ is the state of indistinguishable particles closest to $\Omega_{ext}(\rho)$, then $\frac{1}{2}\lVert \Omegaext(\psi_{ind}) - \Omegaext(\rho)\rVert_{1}=1-F_{ind}(\rho)$. So the task of quantifying indistinguishability reduces to measuring $F_{ind}$, for which \cite{OptimalInterferometricCertification} offers a randomized protocol which can estimate it up to additive error. For the case of pure states such a quantity reduces to 
\begin{equation}
    F_{ind}(\psi)=\frac{\operatorname{perm}(G)}{n!}.
\end{equation}
Thus is reasonable to ask how much the estimation error on $G$ affects such a measure.
\begin{proposition}
    Let $\lVert G-G'\rVert_{2}\leq \epsilon$ then we have that 
    \begin{equation*}
        \left|\frac{\operatorname{perm}(G)}{n!}-\frac{\operatorname{perm}(G')}{n!}\right|\leq \frac{1}{2}\sqrt{n}\epsilon
    \end{equation*}
\end{proposition}
\begin{proof}
    Starting by the relation between Fidelity and trace distance discussed above, we have that
    \begin{equation}
        \left|\frac{\operatorname{perm}(G)}{n!}-\frac{\operatorname{perm}(G')}{n!}\right|=\left|F_{ind}(\psi)-F_{ind}(\psi')\right|=\frac{1}{2}\lVert\Omega_{ext}(\psi)-\Omega_{ext}(\psi')\rVert_{1}\leq \frac{\sqrt{n}}{2}\lVert G-G'\rVert_{2}
    \end{equation}
\end{proof}
So reconstructing the Gram matrix, is sufficient to obtain (up to additive error $\sqrt{n}\epsilon$) an estimate of how close $\Omegaext(\psi)$ is to a state of indistinguishable particles. It is important to clarify that although permanents are involved in the calculation of $F_{ind}(\psi)$, it does not have an exponential complexity in terms of computation up to additive error. In particular one can rewrite it as 
\begin{equation}
    \frac{\operatorname{perm}(G)}{n!}=\frac{1}{n!}\sum_{\sigma \in S_{n}}\prod_{i=1}^{n}G_{i,\sigma_{i}}=\mathbb{E}\left[\prod_{i=1}^{n}G_{i,\sigma_{i}}\right]_{\sigma \in S_{n}}
\end{equation}
which can be understood as an expectation value, for which we can apply Hoeffding's inequality. This result in the fact that $F_{ind}(\psi)$ can be computed up to error $\epsilon_{1}$ and success probability $1-\delta$ in $\mathcal{O}(n\epsilon_{1}^{-2}\ln\delta^{-1})$, where the $n$ factor comes from the computation of $\prod_{i=1}^{n}G_{i,\sigma_{i}}$. By choosing $\epsilon_{1}\ll \sqrt{n}\epsilon$, we can guarantee that 
\begin{equation}
    \left|\frac{\operatorname{perm}(G')}{n!}-F_{ind}(\psi)\right|\leq \frac{1}{2}\sqrt{n}\epsilon
\end{equation}
which allows us to certify how indistinguishable the state is.

\section{Reducing the size of the system}\label{sec: size reduction}

The inversion of the full third-order linear system introduced in Sec.~\ref{sec:relational-info} involves all third-order Bargmann invariants,
\begin{equation}
    A_3 x_3=b_3,
\end{equation}
where the entries of $x_3$ are the real and imaginary parts of $G_{rs}G_{st}G_{tr} \ ,\ 1\leq r<s<t\leq n$. While the full set of third-order invariants is sufficient for the reconstruction, the Gram-matrix reconstruction procedure of Sec.~\ref{sec:protocol-main} only requires the Bargmann invariants associated with paths in a chosen spanning tree. This motivates restricting the inversion to the corresponding subset of parameters.

Let $T$ be the optimal spanning tree of the graph $\Gamma$ introduced in Sec.~\ref{sec:tree}. We partition the vector of third-order invariants as

\begin{equation}
    x_3=
\begin{pmatrix}
x_T\\
x_R
\end{pmatrix},
\end{equation}
where $x_T$ contains the third-order invariants that we retain for the tree-based reconstruction and $x_R$ contains all remaining third-order invariants. Correspondingly, we write
\begin{equation}
    A_3=\begin{pmatrix}A_T & A_R\end{pmatrix},
\end{equation}
so that the measurement equations take the form

\begin{equation}
    A_Tx_T+A_Rx_R=b_3.
\end{equation}
Here $A_T$ and $A_R$ are formed simply by selecting the columns of $A_3$ corresponding to $x_T$ and $x_R$, respectively. The quantities in $x_R$ are nuisance parameters for the purpose of reconstructing the tree-associated invariants. To eliminate their contribution, we use the orthogonal projector onto the complement of the column space of $A_R$,

\begin{equation}
    P_R^\perp=I-A_R A_R^{+},
\end{equation}
where $A_R^{+}$ denotes the Moore--Penrose pseudoinverse. By construction, we have $P_R^\perp A_R=0$. Applying this projector to the measurement equation therefore gives
\begin{equation}
    P_R^\perp A_Tx_T=P_R^\perp b_3.
\end{equation}

Thus the unwanted third-order invariants are removed without requiring their explicit reconstruction. Defining
\begin{equation}
    A_T^{\mathrm{eff}}:=P_R^\perp A_T, \qquad b_T^{\mathrm{eff}}:=P_R^\perp b_3,
\end{equation}

we obtain the reduced linear system $A_T^{\mathrm{eff}}x_T=b_T^{\mathrm{eff}}$. The tree-associated Bargmann invariants can then be estimated as
$\widehat{x}_T=\left(A_T^{\mathrm{eff}}\right)^{+}b_T^{\mathrm{eff}}$.

The conditioning of the reduced inversion is governed by

\begin{equation}
\kappa\!\left(A_T^{\mathrm{eff}}\right)=\frac{\sigma_{\max}\!\left(A_T^{\mathrm{eff}}\right)}
{\sigma_{\min}\!\left(A_T^{\mathrm{eff}}\right)},
\end{equation}
provided that $A_T^{\mathrm{eff}}$ has full column rank. This quantity should be compared directly with the condition number of the full matrix $A_3$. In particular, projecting does not necessarily worsen the conditioning: if the poorly conditioned directions of the full system are associated predominantly with the discarded invariants, the effective matrix $A_T^{\mathrm{eff}}$ can be better conditioned than $A_3$.

\end{document}